\documentclass[12pt,a4paper,twoside]{article}

\usepackage{amsmath,amssymb,amscd}
\usepackage{color}
\numberwithin{equation}{section} 

\newtheorem{theorem}{Theorem}[section]
\newtheorem{lemma}[theorem]{Lemma}
\newtheorem{proposition}[theorem]{Proposition}
\newtheorem{definition}[theorem]{Definition}
\newtheorem{example}[theorem]{Example}

\newenvironment{proof}{\paragraph{Proof.}}{\hfill $\square$\\}
\newenvironment{proof**}[1]{\paragraph{Proof of #1.}}{\hfill $\square$\\}

\newcommand{\e}{\mathcal{E}}
\newcommand{\la}{\lambda}
\newcommand{\bu}{\boldsymbol{u}}
\newcommand{\hk}{\hslash}
\newcommand{\arrow}{\rightarrow}
\newcommand{\map}{\mapsto}
\newcommand{\bb}[1]{\mathbb{#1}}
\newcommand{\Diff}[3]{\left . \frac{d#1}{d#2}\right |_{#3}}
\newcommand{\alg}{\mathfrak{g}}
\newcommand{\Alg}{\mathcal{A}}
\newcommand{\calg}{\widehat{\mathfrak{g}}}
\newcommand{\ca}{\widetilde{\mathfrak{g}}}
\newcommand{\Si}{\mathbb{S}^1}
\newcommand{\Cm}{\mathbb{C}}
\newcommand{\Z}{\mathbb{Z}}
\newcommand{\C}{\mathcal{C}}
\newcommand{\smf}{\mathcal{C}^\infty}
\newcommand{\Rm}{\mathbb{R}}
\newcommand{\Km}{\mathbb{K}}
\newcommand{\F}{\mathcal{F}}
\newcommand{\V}{\mathcal{V}}
\newcommand{\U}{\mathcal{U}}
\newcommand{\pr}{\partial}
\newcommand{\me}{\geqslant}

\newcommand{\bra}[1]{\left (#1\right )}
\newcommand{\brac}[1]{\left [#1\right ]}
\newcommand{\bc}[1]{[\![#1]\!]}
\newcommand{\pobr}[1]{\left \{#1\right \}}
\newcommand{\dual}[1]{\left \langle #1 \right \rangle}
\newcommand{\var}[2]{\frac{\delta #1}{\delta #2}}
\newcommand{\pd}[2]{\frac{\partial #1}{\partial #2}}
\newcommand{\ad}{{\rm ad}}
\newcommand{\cad}{{\rm \widehat{ad}}}
\newcommand{\res}{{\rm res}}
\newcommand{\free}{{\rm free}}
\newcommand{\pmatrx}[2]{\left(\!\!\begin{array}{#1} #2 \end{array}\!\!\!\right)}

\newcommand{\der}[1]{\frac{d}{d #1}}
\newcommand{\Vect}{\rm Vect}
\newcommand{\uf}{{\bf u}}
\newcommand{\T}{{\rm T}}
\newcommand{\cf}{{\bf c}}
\newcommand{\Tr}{{\rm Tr}}
\newcommand{\tr}{{\rm tr}}

\begin{document}

\markboth{M. B\l aszak and B. M. Szablikowski}{Classical $R$-matrix
theory for bi-Hamiltonian field systems}

\title{\bf Classical $R$-matrix theory for bi-Hamiltonian field systems}

\author{Maciej B\l aszak$^{1,}$\footnote{E-mail: {\tt blaszakm@amu.edu.pl}} and
B\l a\.zej M. Szablikowski$^{2,1,}$\footnote{E-mail: {\tt b.szablikowski@maths.gla.ac.uk}}\\[3mm]
\small $^1$ Department of Physics, Adam Mickiewicz University\\
\small Umultowska 85, 61-614 Pozna\'n, Poland\\[2mm]
\small $^2$ Department of Mathematics, University of Glasgow,\\
\small Glasgow G12 8QW, U.K.}

\date{}

\maketitle

\begin{abstract}
The $R$-matrix formalism for the construction of integrable systems with infinitely many degrees of freedom is reviewed. Its application to Poisson, noncommutative and loop algebras as well
 as central extension procedure are presented. The theory is developed for $(1+1)$-dimensional case where the space variable belongs either to $\Rm$ or to various discrete sets. Then, the extension onto $(2+1)$-dimensional case is made, when the second space variable belongs
to $\Rm$. The formalism presented contains many proofs and important
details to make it self-contained and complete. The general theory is applied to several infinite dimensional Lie algebras in order to construct both dispersionless and dispersive (soliton) integrable field systems.
\end{abstract}

\tableofcontents

\section{Introduction}

One of the most characteristic features of integrable nonlinear
systems is the existence of bi-Hamiltonian structures. This ingenious
concept was introduced by F. Magri \cite{Magri} in 1978. From
the geometrical point of view, this means that there exists a
pair of compatible Poisson tensors which allow, using a
recursion chain, to generate infinite (in the infinite-dimensional
case) hierarchies of commuting symmetries and constants of motion
being in involution with respect to the above Poisson tensors. In
order to stress the importance of the bi-Hamiltonian structures for
evolution systems let us quote L.A.~Dickey\footnote{\cite{Dickey}, page 43.}
\begin{quote}
    "The existence of two compatible Poisson (or Hamiltonian) structures is a
    remarkable feature of the most, if not all, integrable systems, sometimes
    it is considered as the essence of the integrability."
\end{quote}

Finding a systematic method for construction of integrable systems
is one of the most important issues in the theory of evolutionary
systems. It is well known that a very powerful tool called the
classical R-matrix formalism \cite{Semenov} proved to be very
fruitful in the systematic construction of the field and lattice
soliton systems as well as dispersionless systems. The crucial point
of the formalism is the observation that integrable dynamical
systems can be obtained from the Lax equations on appropriate Lie
algebras. The greatest advantage of this formalism, besides the
possibility of systematic construction of the integrable systems, is
the construction of bi-Hamiltonian structures and (infinite)
hierarchies of symmetries and conserved quantities.

The goal of this article is to present the formalism of classical
$R$-matrices applied to infinite-dimensional Lie algebras in order
to construct integrable systems with infinitely many degrees of
freedom and related tensor invariants like conserved functions,
symmetries, Poisson tensors, recursion operators and so on.

In the first part of the paper we present in a systematic fashion
the concept of $R$-matrix formalism with many proofs and important
details to make the text self-contained and complete.
First of all we present the basics of the formalism and then illustrate
its approach to the Poisson and noncommutative algebras, with the aim
of construction of integrable hierarchies with multi-Hamiltonian structures.
This theory is developed for (1+1)-dimensional systems, where one variable is the
evolution parameter $t$ (time), which in presented construction
always belongs to $\Rm$, and the spatial variable, which may belong to
$\Km=\Rm$ or $\Cm$ but also to various sets of discrete numbers, so we cover
the cases like continuous dynamics, lattice dynamics and
$q$-dynamics in a single formalism.
Then, we show how to extend the whole formalism over additional
spatial variable via the so-called central extension approach.
The advantage of this construction is
the possibility, in principle, of finding a (2+1)-dimensional counterpart for an arbitrary (1+1)-dimensional system derived via the $R$-matrix
formalism, with the only limitation that the additional space
coordinate is always continuous (belongs to $\Rm$ or $\Si$)
rather than discrete. Finally we apply the $R$-matrix
formalism to loop algebras together with central extension procedure.

In the second part of the paper we apply the formalism developed in
the first part to several important Lie algebras. Some of them allow
to construct dispersionless integrable systems while the others lead
to dispersive (soliton) integrable systems. In the second case we
construct integrable systems with continuous number of degrees of
freedom as well as those with discrete number of degrees of freedom.
As each algebra leads to a huge number of specific integrable
systems, in this paper we present explicitly only few of them,
referring the reader to the cited references for further
examples.

In the following review classical $R$-matrix theory is presented in the framework
of an appropriate Lie algebras, with the aim of construction of integrable bi-hamiltonian systems.
However, the formalism is significantly more powerful as it is intimately connected
with factorization and Riemann-Hilbert problem for the related Lie groups~\cite{Semenov}.

\section{Classical $R$-matrix theory}\label{crf}

In this chapter we will present a unified approach to the
construction of integrable evolution equations together with their
(multi-)Hamiltonian structures. The idea originates from the
pioneering article \cite{Gelfand} by Gelfand and Dickey, where they
presented a construction of Hamiltonian soliton systems of KdV type
using pseudo-differential operators. Next, Adler \cite{Adler} showed
how to construct the bi-Hamiltonian structures for the above soliton
systems using the method based on the Kostant-Symes theorem obtained
independently in \cite{Kost} and \cite{Sym}. Later the abstract
formalism of classical $R$-matrices appeared in the paper \cite{Sk}
by Sklyanin as an intermediate step within the inverse quantum
scattering method. The present version of the formalism was given in
\cite{Reyman,Semenov}. In \cite{Li1,Oevel1} it was shown that there
are in fact three natural Poisson brackets associated with classical
$R$-structures. Quite recently Li \cite{Li} considered the classical
$R$-matrix theory on the so-called (commutative) Poisson algebras.
This approach leads to the construction of multi-Hamiltonian systems
of hydrodynamic (dispersionless) type.

\subsection{Classical $R$-matrices}

Let $\alg$ be a Lie algebra over the field $\Km$ of complex or real
numbers, $\Km =\Cm$ or $\Rm$, that is, $\alg$ is equipped with a
bilinear operation $[\cdot, \cdot]: \alg \times \alg \arrow \alg$,
called a Lie bracket, which is skew-symmetric and satisfies
the Jacobi identity. The Lie bracket $\brac{\cdot, \cdot}$ defines
the adjoint action of $\alg$ on $\alg$: $\ad_a b \equiv \brac{a, b}$.

\begin{definition}
A linear map $R: \alg \arrow \alg$ such that the operation
\begin{equation}\label{rbra}
  \brac{a,b}_R := \brac{R a, b} + \brac{a, R b}\qquad a,b\in\alg
\end{equation}
defines another Lie bracket on $\alg$ is called the classical
$R$-matrix.
\end{definition}

The skew-symmetry of \eqref{rbra} is obvious. As for the Jacobi
identity for \eqref{rbra}, we find that
\begin{align}
\nonumber 0 &= \brac{a,\brac{b, c}_R}_R + {\rm c.p.}
= \brac{Ra, \brac{Rb, c}} + \brac{Ra, \brac{b, Rc}} + \brac{a, R\brac{b, c}_R} + {\rm c.p.}\\
\nonumber &= \brac{Rb, \brac{Rc, a}} + \brac{Rc, \brac{a, Rb}} +
\brac{a, R\brac{b, c}_R} + {\rm c.p.}\\
\label{jacr} &= \brac{a, R\brac{b ,c}_R - \brac{Rb, Rc}} + {\rm c.p.},
\end{align}
where c.p. stands for cyclic permutations within the triple $\{a,b,c\}\in\alg$, and
the last equality follows from the Jacobi identity for $\brac{\cdot,
\cdot}$. Hence, a sufficient condition for $R$ to be a classical
$R$-matrix is to satisfy the so-called (modified) Yang-Baxter
equation, YB($\alpha$):
\begin{equation}\label{YB}
  \brac{R a, R b} -R \brac{a, b}_R + \alpha \brac{a,b} = 0,
\end{equation}
where $\alpha$ is a number from $\Km$. There are only two relevant
cases of YB($\alpha$), namely $\alpha\neq 0$ and $\alpha = 0$, as
all Yang-Baxter equations with $\alpha\neq 0$ are equivalent up to a
rescaling of $\alpha$.

\begin{definition}
A linear operator $A:\alg\arrow\alg$ is called intertwining if
$A\circ\ad_a = \ad_a\circ A$, i.e., if $A\brac{a,b} = \brac{A a,b} = \brac{a, A b}$
for any $a,b\in\alg$.
\end{definition}

\begin{proposition}\cite{rs}\label{int}
If $R$ is a classical $R$-matrix and $A$ an intertwining operator,
then $R\circ A$ also is a classical $R$-matrix.
\end{proposition}
\begin{proof}
We have
\begin{align*}
    RA\brac{a,b}_{RA} &= RA\brac{RAa,b} + RA\brac{a,RAb}\\
     &= R\brac{RAa,Ab} + R\brac{Aa,RAb} = R\brac{Aa,Ab}_R.
\end{align*}
Hence,
\begin{equation*}
    \brac{RAa,RAb} - RA\brac{a,b}_{RA} = \brac{RAa,RAb} - R\brac{Aa,Ab}_R
\end{equation*}
and the Jacobi identity for $[\cdot,\cdot]_{RA}$ with respect to the
elements $a,b,c\in\alg$ reduces to the Jacobi identity for
$[\cdot,\cdot]_R$ with respect to the elements $Aa,Ab,Ac$, see
\eqref{jacr}.
\end{proof}

Notice that a linear combination of intertwining operators again is
an intertwining operator.

\subsection{Lax hierarchy}\label{seclax}

In this section we present the classical $R$-matrix formalism for
the class of Lie algebras for which the Lie bracket additionally
satisfies the Leibniz rule. Later, while
considering the loop algebras in Section \ref{la}, we shall drop this extra condition.

Assume that the Lie algebra $\alg$ is also an algebra with respect
to an associative multiplication such that
\begin{equation}\label{derm}
  \ad_a(bc) = \ad_a(b)c + b\ad_a(c)\quad\iff\quad\brac{a, b c} = \brac{a, b}c + b\brac{a, c}
\end{equation}
the Leibniz rule holds, i.e., the Lie bracket $[\cdot, \cdot]$ is a
derivation with respect to the multiplication. Notice that this
condition is satisfied automatically in the case of a commutative
algebra $\alg$ when the Lie bracket is given by a finite-dimensional
Poisson bracket, as well as in the case of a non-commutative algebra
$\alg$ with the Lie bracket given by the commutator.

As a consequence of the above assumption, any well-defined smooth
mapping $X: \alg \arrow \alg$, $L\map X(L)$,
is an invariant of the Lie bracket, i.e.,
\begin{equation}\label{inv}
  \brac{X(L), L} = 0.
\end{equation}
The smoothness of the mapping $X$ means that its differential and
directional derivatives exist and are well-defined.

\begin{proposition}\label{ybp}
Smooth invariant functions $X_n(L)$ generate a hierarchy of vector
fields on $\alg$ of the form
\begin{equation}\label{lax}
  L_{t_n} = \brac{R X_n(L), L}\qquad L\in\alg\quad n\in\Z,
\end{equation}
where $t_n$ are evolution parameters. Assuming that a classical
$R$-matrix $R$ commutes with the directional derivatives with respect to
all \eqref{lax}, the Yang-Baxter equation \eqref{YB} is a sufficient
condition for the pairwise commutativity of the vector fields
\eqref{lax}.
\end{proposition}
\begin{proof}
The directional derivative of a smooth function $F:\alg \arrow \alg$
in the direction of \eqref{lax} is given by
\begin{equation}\label{dirF}
  F(L)_{t_n} = F(L)'\brac{L_{t_n}} = \brac{R X_n(L), F(L)},
\end{equation}
which follows from the Leibniz rule \eqref{derm}. Thus, one finds that
\begin{align}
\nonumber   &\bra{L_{t_m}}_{t_n} - \bra{L_{t_n}}_{t_m} =
\brac{RX_m(L),L}_{t_n} - \brac{RX_n(L),L}_{t_m} \\
\nonumber  &\qquad = \brac{\bra{RX_m(L)}_{t_n} - \bra{RX_n(L)}_{t_m}, L} +
\brac{RX_m(L),\brac{RX_n(L),L}} -
\brac{RX_n(L),\brac{RX_m(L),L}}\\
 \label{comm1} &\qquad = \brac{\bra{RX_m(L)}_{t_n} - \bra{RX_n(L)}_{t_m} +
\brac{RX_m(L),RX_n(L)},L}.
\end{align}
Hence, for the pairwise commutativity of vector fields \eqref{lax} it suffices
that the so-called zero-curvature equations
\begin{equation}\label{zero}
\bra{RX_m(L)}_{t_n} - \bra{RX_n(L)}_{t_m} + \brac{RX_m(L),RX_n(L)} =
0
\end{equation}
hold.

The assumption that $R$ commutes with directional derivatives implies
that it commutes with the derivatives with respect to evolution
parameters, i.e.,
\begin{equation}\label{assum}
    \bra{R L}_{t_n} = R L_{t_n}.
\end{equation}
Thus, the right hand side of \eqref{zero} becomes
\begin{align}
\nonumber &R\bra{X_m(L)}_{t_n} - R\bra{X_n(L)}_{t_m} + \brac{RX_m(L),RX_n(L)}=\\
\nonumber &\qquad \overset{{\rm by\,\eqref{dirF}}}{=}
R\brac{RX_n(L),X_m(L)} - R\brac{RX_m(L),X_n(L)} + \brac{RX_m(L),RX_n(L)}\\
\label{comm2} &\qquad = \brac{RX_m(L),RX_n(L)} - R\brac{X_m(L),X_n(L)}_R.
\end{align}
Now, if the $R$-matrix $R$ satisfies the Yang-Baxter equation
$\eqref{YB}$ then the last expression is equal to $-\alpha
\brac{X_m(L),X_n(L)} = 0$, and the result follows.
\end{proof}

The hierarchy \eqref{lax} is called {\it the Lax hierarchy} and $L$
is called the Lax operator or the Lax function depending on the
nature of a given Lie algebra $\alg$. Notice that the assumption
that $R$ commutes with directional derivatives is an important
condition although is not enunciated explicitly in most works on the
$R$-matrices.

The natural choice for invariant smooth functions are the power
functions $X_n(L)=L^n$ that are always well-defined on $\alg$. One
can consider less trivial functions, for example the logarithmic
ones, like $X(L) = \ln L$, but only when they have proper
interpretation in $\alg$.

It is natural to ask when the abstract Lax hierarchy \eqref{lax}
represents a ``real" hierarchy of integrable evolution systems on a
suitable function space constituting an infinite-dimensional phase
space $\U$. This occurs when we can construct an embedding map
$\iota:\U\arrow\alg$, which induces the differential structure on
$\alg$. Notice that for $\iota$ being an embedding its differential
$\iota':\V\arrow\alg$ is an injective map, where $\V$ is a space of
vector fields on $\U$. In such a case the Lax hierarchy \eqref{lax}
can be pulled back to the original function space by $\iota'^{-1}$.
The symmetries from the Lax hierarchy \eqref{lax} represent
compatible evolution systems when the left- and right-hand sides of
\eqref{lax} span the same subspace of $\alg$. So, the Lax element
$L$ of $\alg$ has to be chosen in a suitable fashion.

\subsection{Simplest $R$-matrices}\label{scheme}

The simplest way to obtain a classical $R$-matrix is to decompose a
given Lie algebra into Lie subalgebras.  Thus, assume that a Lie
algebra $\alg$ can be split into a (vector) direct sum of Lie subalgebras $\alg_+$ and $\alg_-$,
i.e.,
\begin{equation*}
\alg = \alg_+ \oplus \alg_-\qquad \brac{\alg_\pm, \alg_\pm}\subset
\alg_\pm\qquad \alg_+\cap \alg_- = \emptyset .
\end{equation*}
It is important to stress that we do not require that $[\alg_+,\alg_-]=0$.

Upon denoting the projections onto the subalgebras in question by
$P_\pm$, we define a linear map $R: \alg\arrow \alg$ as
\begin{equation}\label{rp}
R = \frac{1}{2} (P_+ - P_-).
\end{equation}
Using the equality $P_+ + P_- = 1$ \eqref{rp} can be represented in
the following equivalent forms:
\begin{equation}\label{rp0}
R = P_+ - \frac{1}{2} = \frac{1}{2} - P_- .
\end{equation}

Let $a_\pm:= P_\pm (a)$ for $a\in \alg$. Then
\begin{equation*}
  \brac{a, b}_R = \brac{a_+, b_+} - \brac{a_-, b_-}\quad
  \Longrightarrow\quad
  R\brac{a, b}_R = \frac{1}{2} \brac{a_+, b_+} + \frac{1}{2} \brac{a_-,
  b_-}
\end{equation*}
and
\begin{equation*}
  \brac{Ra, Rb} = \frac{1}{4}\brac{a_+, b_+} - \frac{1}{4}\brac{a_+,
  b_-} - \frac{1}{4}\brac{a_-, b_+} + \frac{1}{4}\brac{a_-, b_-}.
\end{equation*}
Hence, the map \eqref{rp} satisfies the Yang-Baxter equation
\eqref{YB} for $\alpha = \frac{1}{4}$ and is a well-defined
classical $R$-matrix. This is the simplest and the most common
example of a well-defined $R$-matrix.

For instance, the Lax hierarchy \eqref{lax} for the $R$-matrix
\eqref{rp}, following from the decomposition of a Lie algebra into
Lie subalgebras, takes the form
\begin{equation*}
L_{t_n} = \brac{\bra{X_n(L)}_+, L} = - \brac{\bra{X_n(L)}_-,L}.
\end{equation*}
It can be written in two equivalent ways because \eqref{rp0} holds.

The construction of majority of known integrable systems within
the formalism presented above is based on the classical $R$-matrices
that follow from the double decomposition (\ref{rp}) of Lie algebras
into Lie subalgebras. In \cite{Oevel4,Szablikowski1,skryp1,Serg} the authors considered deformations of
\eqref{rp} that preserve the Yang-Baxter equation and originate from
a triple decomposition of a given Lie algebra; see also
\cite{szabla} for multiple decompositions of Lie algebras.

\subsection{Lie-Poisson structures}

Let $\alg^*$ be a (regular) dual of a given Lie algebra $\alg$ and
$\dual{\cdot, \cdot}:\alg^*\times \alg\arrow \Km$ be the usual
duality pairing. The co-adjoint action $\ad^*$ of $\alg$ on $\alg^*$
is defined through the relation
\begin{equation}\label{coad}
\dual{\ad^*_a \eta, b} + \dual{\eta, \ad_a b} = 0\quad \iff\quad
\dual{\ad^*_a\eta, b} = -\dual{\eta, \brac{a, b}},
\end{equation}
where $a,b\in\alg$ and $\eta\in \alg^*$

Let this time $\iota:\U\arrow \alg^*$ be the embedding of the
original phase space into the dual Lie algebra. Then every
functional $F:\U\arrow \Km$ can be extended to a smooth function on
$\alg^*$. Therefore, let $\smf(\alg^*)$ be the space of all smooth
functions on $\alg^*$ of the form $F\circ
\iota^{-1}:\alg^\star\arrow \Km$, where $F\in \smf(\U)$. Then the
differentials $dF(\eta)$ of $F(\eta)\in \smf(\alg^*)$ at the point
$\eta\in \alg^*$ belong to $\alg$ as they can be evaluated using the
relation
\begin{equation}\label{grad}
F(\eta)'[\xi]\equiv \Diff{F(\eta+\epsilon \xi)}{\epsilon}{\epsilon=0} =
\dual{\xi, dF(\eta)}\qquad \xi\in \alg^*\quad \epsilon\in\Km.
\end{equation}
Moreover, the form of differentials $dF\in\alg^*$ has to be such
that the duality pairing between $\alg$ and its dual $\alg^*$
coincides with the duality map between vector fields and one-forms
on the original infinite-dimensional function phase space $\U$.
Indeed,
\begin{equation}\label{eu}
  \dual{\eta_t,dF} = \int\sum_{i=0}^\infty \var{F}{u_i} (u_i)_t\,dx,
\end{equation}
where $\eta_t\in\alg^*$ is a vector field on $\alg^*$,
$F(\eta)\in\smf(\alg^*)$ is a functional depending on the dynamical
fields $u_i$ from the phase space $\U$, $\var{F}{u_i}$ is the
variational derivative of $F$ with respect to field variable $u_i$
and in the case of discrete functionals the symbol of integration in
\eqref{eu} must be replaced by an appropriate summation.
We also have the relation
\begin{equation}\label{dd2}
\dual{\xi, dF'[\eta]} = \dual{\eta, dF'[\xi]},
\end{equation}
which is equivalent to the vanishing of the square of the
exterior differential, i.e., $d^2F=0$.

We also make an additional assumption that the Lie bracket in $\alg$
is such that directional derivative along an arbitrary
$\xi\in\alg^\star$ is a derivation of the Lie bracket. This means
that the following relation holds:
\begin{align}\label{cond}
\brac{a, b}'[\xi] = \brac{a'[\xi], b} + \brac{a, b'[\xi]}.
\end{align}

\begin{theorem}\label{nlp}
There exists a Poisson bracket on the space of smooth functions on a
dual algebra $\alg^*$, which is induced by the Lie bracket on
$\alg$. This Poisson bracket is defined as follows:
\begin{equation}\label{nliepo}
\pobr{H,F}(\eta):= \dual{\eta, \brac{dF, dH}}\qquad \eta\in \alg^*
\quad H, F\in \smf(\alg^*).
\end{equation}
\end{theorem}

\begin{lemma}
The differential of \eqref{nliepo} is given by
\begin{equation}\label{dn}
d\pobr{H, F} = \brac{dF, dH} + dF'\brac{\ad^*_{dH}\eta} -
dH'\brac{\ad^*_{dF}\eta}.
\end{equation}
\end{lemma}
\begin{proof}
By \eqref{grad} we find that
\begin{align*}
\pobr{H, F}'[\xi] &= \dual{\eta'[\xi], \brac{dF, dH}}+\dual{\eta,
\brac{dF'[\xi], dH}+\brac{dF, dH'[\xi]}}\\
&\overset{{\rm by\,\eqref{coad}}}{=} \dual{\xi, \brac{dF, dH}}
+\dual{\ad^*_{dH}\eta, dF'[\xi]}-\dual{\ad^*_{dF}\eta, dH'[\xi]}\\
&\overset{{\rm by\,\eqref{dd2}}}{=} \dual{\xi, \brac{dF, dH} +
dF'\brac{\ad^*_{dH}\eta} - dH'\brac{\ad^*_{dF}\eta}},
\end{align*}
and the result follows.
\end{proof}

\begin{proof**}{Theorem \ref{nlp}}
Bilinearity and skew-symmetry of \eqref{nliepo} are obvious, so we
only have to prove the Jacobi identity:
\begin{align*}
& \pobr{F,\pobr{G,H}} + c.p. = \dual{\eta, \brac{d\pobr{G, H}, dH} +
c.p.}\\
&\qquad\quad\overset{{\rm by\,\eqref{dn}}}{=} \dual{\eta,
\brac{\brac{dH, dG}, dF}
+\brac{dH'\brac{\ad^*_{dG}\eta}, dF} - \brac{dG'\brac{\ad^*_{dH}\eta}, dF} + c.p.}\\
&\qquad\quad\overset{{\rm by\,\eqref{coad}}}{=} \dual{\eta,
\brac{\brac{dH,
dG}, dF}} + \dual{\ad^*_{dF} \eta, dH'\brac{\ad^*_{dG}\eta}} - \dual{\ad^*_{dF}\eta, dG'\brac{\ad^*_{dH}\eta}} + c.p.\\
&\qquad\quad\overset{{\rm by\,c.p.}}{=}\dual{\eta, \brac{\brac{dH,
dG}, dF}} + \dual{\ad^*_{dF}\eta,
dH'\brac{\ad^*_{dG}\eta}} - \dual{\ad^*_{dG}\eta, dH'\brac{\ad^*_{dF}\eta}} + c.p.\\
&\qquad\quad\overset{{\rm by\,\eqref{dd2}}}{=} \dual{\eta,
\brac{\brac{dH, dG}, dF} + c.p.} = 0,
\end{align*}
where the last equality follows from the Jacobi identity for
$\brac{\cdot, \cdot}$.
\end{proof**}

The bracket \eqref{nliepo} is called a (natural) {\it Lie-Poisson bracket}
and was originally discovered by Sophus Lie. Its modern formulation
is due to Berezin \cite{Ber} as well as Kirillov and Kostant
\cite{Kir}.\looseness=-1

Now assume that we have an additional Lie bracket \eqref{rbra} on
$\alg$ defined through the classical $R$-matrix such that
\eqref{assum} is valid. Then \eqref{rbra} satisfies the condition
\eqref{cond}. As a result, there is another well-defined (by Theorem
\ref{nlp}) Lie-Poisson bracket on the space of scalar fields
$\smf(\alg^*)$:
\begin{equation}\label{rliepo}
\pobr{H,F}_R(\eta):= \dual{\eta, \brac{dF, dH}_R}\qquad \eta\in
\alg^* \quad H, F\in \smf(\alg^*).
\end{equation}

Using \eqref{coad} we find that the associated Poisson operators at
$\eta\in \alg^*$, the one for the natural Lie-Poisson bracket
\eqref{nliepo} and the second one, \eqref{rliepo}, have the form
\begin{align*}
&\pobr{H, F} = \dual{\pi dH, dF}\quad\iff\quad \pi: dH\map \ad^*_{dH}\eta\\
&\pobr{H, F}_R = \dual{\pi_R dH, dF} \quad\iff\quad \pi_R: dH \map
\ad^*_{R dH}\eta + R^* \ad^*_{dH}\eta,
\end{align*}
where the adjoint of $R$ is defined by the relation $\dual{R^*\eta, a} = \dual{\eta, R a}$, where
$\eta\in \alg^*$ and $a\in\alg$.

The following theorem constitutes the essence of the classical
$R$-matrix formalism.

\begin{theorem}[\cite{Semenov}]\label{rsts}
The Casimir functions $C_n\in\smf(\alg^*)$ of the natural
Lie-Poisson bracket \eqref{nliepo} are in involution with respect to
the Lie-Poisson bracket \eqref{rliepo} induced by \eqref{rbra}.
Moreover, $C_n$ generate a hierarchy of vector fields on $\alg^*$:
\begin{equation}\label{duallax}
  \eta_{t_n} = \pi_R dC_n(\eta) = \ad^*_{R dC_n}\eta\qquad \eta\in\alg^*.
\end{equation}
The evolution systems \eqref{duallax} pairwise commute, i.e.,
$(\eta_{t_m})_{t_n} = (\eta_{t_n})_{t_m}$, and are Hamiltonian with
respect to \eqref{rliepo}. Moreover, any equation from
\eqref{duallax} admits all Casimir functions $C_n$ of (\ref{nliepo})
as integrals of motion.
\end{theorem}
\begin{proof}
The Casimir functions $\C_n$ of the natural Lie-Poisson bracket
\eqref{nliepo} satisfy the following condition
\begin{equation*}
\forall\ F\in \smf(\alg^*)\quad  \pobr{F, C_n} = 0\quad \iff\quad
\ad^*_{dC_n}\eta =0,
\end{equation*}
that is, their differentials are $\ad^*$-invariant. Hence, they are
in involution with respect to the Lie-Poisson bracket
\eqref{rliepo}, i.e., $\pobr{C_n, C_m}_R = 0$. Now, as $\pi_R d$ is a Lie algebra homomorphism
from the Poisson algebra of smooth functions on $\alg^*$ to the Lie algebra
 of vector fields on $\alg^*$, commutativity of Hamiltonian vector fields \eqref{duallax}, with the
Casimir functions as Hamiltonians, readily follows.
\end{proof}

In fact, when the $R$-matrix follows from the decomposition of an Lie algebra
into sum of Lie subalgebras, i.e. is given by \eqref{rp}, Theorem \ref{rsts} can be considered
as a generalization of Kostant-Symes theorem \cite{Kost,Sym}.

The construction of Casimir functions $C_n$ and related dynamical
systems \eqref{duallax} on the dual Lie algebra $\alg^*$ is, in
contrast with \eqref{lax}, often inconvenient and impractical. Thus,
a formulation of a similar theory on $\alg$ instead on $\alg^*$ is
often justified. This can be done when one can identify $\alg^*$
with $\alg$ by means of a suitable scalar product.

\subsection{$Ad$-invariant scalar products}

We restrict our further considerations to the Lie algebras $\alg$
for which their duals $\alg^*$ can be identified with $\alg$ through
a duality map. So, we assume the existence of a bi-linear scalar
product
\begin{equation}\label{sym}
\bra{\cdot, \cdot}_\alg: \alg \times \alg \arrow \Km
\end{equation}
on $\alg$, and we assume this product to be symmetric: $\bra{a, b}_\alg = \bra{b, a}_\alg$,
and non-degenerate, that is, $a=0$ is the only element of $\alg$
that satisfies $\bra{a, b}_\alg = 0$ for all $b\in \alg$.
Then we can identify $\alg^*$ with $\alg$ ($\alg^* \cong \alg$) by
setting $\dual{\eta, b} = \bra{c, b}_\alg$, $\forall b\in \alg$,
where $\eta\in \alg^*$ is identified with $c\in \alg$. We also make
an additional assumption that the symmetric product \eqref{sym} is
$ad$-invariant, i.e.,
\begin{equation}\label{adinv}
\bra{\brac{a, c}, b}_\alg + \bra{c, \brac{a, b}}_\alg = 0.
\end{equation}
This is a counterpart of the relation \eqref{coad}. Thus, if
$\eta\in \alg^*$ is identified with $c\in \alg$ we have $\dual{\ad^*_a \eta, b} = \bra{\brac{a, c}, b}_\alg$ and one identifies $\ad^*_a \eta\in \alg^*$ with $\ad_a c\in \alg$.

In fact, by virtue of the scheme presented in the previous section
all equations from the hierarchy \eqref{lax} are Hamiltonian. Since
$\alg^*\cong\alg$, the Lie-Poisson brackets \eqref{nliepo} and
\eqref{rliepo} on the space of scalar fields
$\smf(\alg\cong\alg^*)$ at $L\in \alg$ take the form
\begin{align}
\nonumber  &\pobr{H, F} = \bra{L, \brac{dF, dH}}_\alg = \bra{dF, \pi
dH}_\alg \quad\iff\quad \pi dH=\brac{dH, L}\\
\label{lin}  &\pobr{H, F}_R = \bra{L, \brac{dF, dH}_R}_\alg = \bra{dF,\pi_R dH}_\alg
\ \iff\ \pi_R dH = \brac{R dH, L} + R^*\brac{dH, L},
\end{align}
where now $R^*$ is defined by the relation $\bra{R^*a, b}_\alg = \bra{a, Rb}_\alg$.

Differentials of the Casimir functions $\C_n(L)\in\smf(\alg)$ of the
natural Lie-Poisson bracket are invariants of the Lie bracket, i.e.,
$\brac{dC_n(L), L}= 0$. Obviously, the Casimir functions are still
in involution with respect to the second Lie-Poisson bracket defined
by $R$ and generate pairwise commuting Hamiltonian vector fields of
the form
\begin{equation}\label{hamlax}
  L_{t_n} = \pi_R dC_n(L) = \brac{R dC_n, L}.
\end{equation}

Notice that the Lax hierarchy \eqref{lax} coincides with \eqref{hamlax}
for $X_n = dC_n$. Moreover, it follows that if there exists a
symmetric, non-degenerate and $ad$-invariant product on $\alg$ then
the Yang-Baxter equation \eqref{YB} is not a necessary condition for
the commutativity of vector fields from the Lax hierarchy
\eqref{lax}. However, if \eqref{YB} is not satisfied then the
zero-curvature equations \eqref{zero} will not be automatically
satisfied as well.

The simplest way to define an appropriate scalar product on a Lie
algebra $\alg$ is to use a trace form $\Tr: \alg \arrow \bb{K}$ such
that the scalar product
\begin{equation}\label{sc}
\bra{a,b}_\alg := \Tr \bra{ab}\qquad a,b\in \alg
\end{equation}
is nondegenerate. In this case the symmetry of \eqref{sc} entails
that
\begin{equation}\label{abba}
  \Tr \bra{a b} = \Tr \bra{b a}.
\end{equation}
\begin{lemma}\label{adtr}
Let $\Tr: \alg \arrow \bb{K}$ be a trace form defining a symmetric
and nondegenerate scalar product \eqref{sc} such that the trace of
Lie bracket vanishes, i.e. $\Tr \brac{a, b}=0$ for all $a,b\in\alg$.
Then the condition \eqref{derm} for the Lie bracket to be a
derivation with respect to the multiplication is a sufficient
condition for \eqref{sc} to be $\ad$-invariant.

Moreover, if the Lie bracket in $\alg$ is given by the commutator,
$[a,b]=ab-ba$, then the $\ad$-invariance follows from the associativity
of the multiplication in $\alg$.
\end{lemma}
\begin{proof}
The first part of the lemma is immediate, as
\begin{align*}
\bra{\brac{a,c},b}_\alg+\bra{c,\brac{a,b}}_\alg = \Tr
\bra{\brac{a,c}b+c\brac{a,b}} = \Tr \brac{a,cb} = 0,
\end{align*}
where we used the assumptions from the proposition. The second
statement of the lemma follows immediately from the assumption and \eqref{abba}.
\end{proof}

Under the assumption that we have an algebra $\alg$ such that a Lie
bracket is a derivation of a multiplication \eqref{derm} and $\alg$
is endowed with a trace form inducing a nondegenerate
$\ad$-invariant scalar product \eqref{sc}, the most natural Casimir
functions $\C_n(L)\in \smf(\alg)$ of the Lie-Poisson bracket
\eqref{nliepo} are given by the traces of powers of $L$, i.e.,
\begin{equation}\label{casq}
C_n(L) = \frac{1}{n+1} \Tr \bra{L^{n+1}}\quad \iff \quad dC_n = L^n
\qquad n\neq -1.
\end{equation}
The associated differentials are found from the expression
\eqref{grad}, which can be now reduced to
\begin{equation}\label{difft}
\der{t}F(L) = \bra{L_t, dF}_\alg = \Tr \bra{L_t\ dF}\qquad L\in
\alg,
\end{equation}
where $t$ is an evolution parameter associated with a vector field
$L_t$ on $\alg$.

\subsection{Hamiltonian structures on Poisson algebras}\label{pa}

\begin{definition}
Let $\Alg$ be a commutative, associative algebra with unit $1$. If
there is a Lie bracket on $\Alg$ such that for each element $a\in
\Alg$ the operator $\ad_a:b\map \{a,b\}$ is a derivation of the
multiplication, i.e. $\{a,bc\} = \{a,b\}c+b\{a,c\}$ then
$(\Alg,\{\cdot,\cdot\})$ is called a Poisson algebra and the bracket
$\pobr{\cdot,\cdot}$ is a Poisson bracket.
\end{definition}
Thus, the Poisson algebras are Lie algebras with the Lie bracket
$[\cdot,\cdot]:=\{\cdot,\cdot\}$ endowed with an additional
structure. Of course we should not confuse the above bracket with
the Poisson brackets in the algebra of scalar fields. It will follow
easily from the context which bracket is used. In the case of the
Poisson algebra $\Alg$ a classical $R$-matrix defines the second Lie
product on $\Alg$ but not the Poisson bracket; in general this would
not be possible.

\begin{theorem}\cite{Li}
Let $\Alg$ be a Poisson algebra with the Poisson bracket
$\{\cdot,\cdot\}$ and a non-degenerate $\ad$-invariant scalar
product $(\cdot,\cdot)_\Alg$ such that the operation of
multiplication is symmetric with respect to the latter, i.e.,
$(ab,c)_\Alg = (a,bc)_\Alg$ for all $a,b,c\in \Alg$. Assume that $R$
is a classical $R$-matrix such that \eqref{assum} holds.

Then for any integer $n\me 0$ the formula
\begin{equation}\label{pobr}
\pobr{H,F}_n = \bra{L, \pobr{R(L^ndF),dH}+\pobr{dF,R(L^ndH)}}_\Alg,
\end{equation}
where $H,F$ are smooth functions on $\Alg$, defines a Poisson
structure on $\Alg$. Moreover, all brackets $\{\cdot,\cdot\}_n$ are
compatible.
\end{theorem}
An important property that classical $R$-matrices commute with
differentials of smooth maps from $\Alg$ to $\Alg$, or equivalently
satisfy \eqref{assum}, is used in the proof of Theorem~4.2 of
\cite{Li}, although it is not explicitly stated there. In fact, the
existence of scalar product being symmetric with respect to the
multiplication, $(ab,c)=(a,bc)$, entails existence of a trace form
on $\Alg$. Setting $c=1$ we have $(ab,1)=(a,b)$. Thus, the trace can
be defined as $\Tr (a) := (a, 1) = (1, a)$.

The Poisson operators $\pi_n$ related to the Poisson brackets
\eqref{pobr} such that $\pobr{H,F}_n = (dF,\pi_ndH)$, are given by
the following Poisson maps
\begin{equation}\label{pot}
\pi_n: dH\map \pobr{R(L^ndH),L} + L^n R^*\bra{\pobr{dH,L}}\qquad
n\me 0.
\end{equation}
Notice that the bracket \eqref{pobr} with $n=0$ is just the
Lie-Poisson bracket with respect to the second Lie bracket on $\Alg$
defined by a classical $R$-matrix. Referring to the dependence on
$L$, the Poisson maps \eqref{pot} are called linear for $n=0$, quadratic
for $n=1$, and cubic for $n=2$, respectively. The Casimir functions
$C(L)$ of the natural Lie-Poisson bracket are in involution with
respect to all Poisson brackets \eqref{pot} and generate pairwise
commuting Hamiltonian vector fields of the form
\begin{equation*}
  L_t = \pi_n dC = \pobr{R \bra{L^n dC}, L}\qquad L\in \Alg.
\end{equation*}

Taking the most natural Casimir functions \eqref{casq}, defined by
traces of powers of $L$, for the Hamiltonians, one finds a hierarchy
of evolution equations which are multi-Hamiltonian dynamical
systems:
\begin{equation}\label{eveq}
L_{t_q} = \pobr{R dC_n,L} = \pi_{0} dC_n = \pi_1 dC_{n-1} = \dots =
\pi_l dC_{n-l} = \dots,
\end{equation}
where $C_n$ are such that $dC_n = LdC_{n-1}$. For any $R$-matrix any
two evolution equations in the hierarchy \eqref{eveq} commute
because of involutivity of the Casimir functions $C_q$. Each
equation admits all the Casimir functions as conserved quantities in
involution. In this sense we will consider \eqref{eveq} as a
hierarchy of integrable evolution equations. The most natural choice
for the Casimir functions are the traces of the power functions
\eqref{casq}.

\subsection{Hamiltonian structures on noncommutative algebras}\label{na}

In this section, in contrast with the previous one, we will consider
a noncommutative associative algebra $\alg$, with unity, for which the Lie
structure is defined as a commutator, i.e., $\brac{a,b}:= a b - b a$, where $a,b\in\alg$.
Such a Lie bracket automatically satisfies the required Leibniz rule
\eqref{derm}. We further assume existence of nondegenerate,
symmetric and $\ad$-invariant scalar product on $\alg$. Let $R$ be a
classical $R$-matrix such that \eqref{assum} is satisfied.

In this case the situation is more involved and only three explicit
forms of Poisson brackets on the space of smooth functions
$\smf(\alg)$ defined by related Poisson tensors are known from the
literature:
\begin{equation*}
\pobr{H,F}_n = \bra{dF,\pi_n dH}_\alg\qquad n=0,1,2.
\end{equation*}
These Poisson brackets (or associated tensors) are called linear,
quadratic and cubic brackets (resp.\ tensors) for $n=0,1,2$,
respectively.

The linear one is simply the Lie-Poisson bracket, with respect to
the second Lie structure on $\alg$ defined by classical $R$-matrix,
with the Poisson tensor \eqref{lin}
\begin{equation}\label{pot1}
\pi_0 dH = \brac{R dH, L} + R^* \brac{dH, L},
\end{equation}
for which there is no need for additional assumptions.

In our further considerations we have to assume that the scalar product
is symmetric with respect to the operation of multiplication,
$(ab,c)=(a,bc)$. Note that this property implies that the scalar
product is automatically $\ad$-invariant with respect to the Lie
bracket defined by the commutator (Lemma \ref{adtr}).

The quadratic case is more delicate. A quadratic tensor \cite{Suris}
\begin{equation}\label{pot2a}
  \pi_1 dH = A_1 (L dH) L - L  A_2 (dH L) + S(dH L)  L - L S^* (L dH)
\end{equation}
defines a Poisson tensor if the linear maps $A_{1,2}: \alg \arrow\alg$
are skew-symmetric, $A_{1,2}^*=-A_{1,2}$, satisfy YB($\alpha$) \eqref{YB} for $\alpha \neq 0$ and the linear
map $S: \alg \arrow \alg$ with its adjoint $S^*$ satisfy
\begin{equation}\label{cc}
S\bra{\brac{ A_2 a, b} + \brac{a,A_2 b}} = \brac{S a, S b},\quad S^* \bra{\brac{ A_1 a, b} + \brac{a,A_1 b}} = \brac{S^*a, S^*b}.
\end{equation}
In the special case when
\begin{equation}\label{rr}
  \widetilde{R} := \frac{1}{2} \bra{R-R^*}
\end{equation}
satisfies the YB($\alpha$), for the same $\alpha$ as $R$, under the
substitution $A_1 = A_2 = R-R^*$ and $S=S^*=R+R^*$
the quadratic Poisson operator \eqref{pot2a} reduces to
\cite{Oevel1,Li1}
\begin{align}\label{pot2b}
  \pi_1 dH = \brac{R \brac{dH,L}_+,L} + L R^*\brac{dH,L} + R^*\bra{\brac{dH,L}}L,
\end{align}
where $\brac{a,b}_+ := ab+ba$ and the conditions \eqref{cc} are equivalent to YB($\alpha$) for $R$
and $\widetilde{R}$. In particular, when $R^*=-R$, these
conditions are automatically satisfied as in this case
$\widetilde{R} = R$.

Another special case occurs when the maps $A_{1,2}$ and $S$
originate from the decomposition of a given classical $R$-matrix
satisfying YB($\alpha$) for $\alpha \neq 0$
\begin{equation*}
  R = \frac{1}{2}\bra{A_1 + S} = \frac{1}{2}\bra{A_2 + S^*},
\end{equation*}
where $A_{1,2}$ are skew-symmetric. Then, the conditions \eqref{cc}
imply that both $A_1$ and $A_2$ satisfy YB($\alpha$) for the same
value of $\alpha$ as $R$, \cite{Oevel4}. Hence, in this case we only
have to check the conditions \eqref{cc} for \eqref{pot2a} to be a
Poisson operator. The latter now takes the form
\begin{align}\label{pot2c}
\pi_1: dH\map 2R\bra{L dH}L-2LR\bra{dH L} + S\bra{\brac{dH,L}}L
  + L S^* \brac{dH,L}.
\end{align}

Finally, the cubic tensor $\pi_2$ takes the simple form
\cite{Oevel1}
\begin{equation*}\label{pot3}
  \pi_2: dH\map \brac{R\bra{L dH L},L} + L R^*\bra{\brac{dH,L}}L
\end{equation*}
and is Poisson without any further additional assumptions.

Once again, taking the Casimir functions \eqref{casq} defined by the
traces of powers of $L$ for the Hamiltonians yields a hierarchy of
evolution equations which are tri-Hamiltonian dynamical systems,
\begin{equation}\label{eveq2}
L_{t_n} = \brac{R dC_n,L} = \pi_{0} dC_n = \pi_1 dC_{n-1} = \pi_2
dC_{n-2},
\end{equation}
where $C_n$ are such that $dC_n = LdC_{n-1}$. We assumed that
$\pi_2$ in \eqref{eveq2} is given by \eqref{pot2b} or \eqref{pot2c}.
In the first case all three Poisson tensors in \eqref{eveq2} are
automatically compatible. In the second case this has to be checked
separately.

\subsection{Central Extension Approach}\label{cea}

Let $\alg$ be a Lie algebra with the Lie bracket $[\cdot,\cdot]$.
Consider its extension
 $\calg:= \alg\oplus\Km$ with the Lie bracket given by
\begin{equation}\label{celb}
 \brac{(a,\alpha),(b,\beta)}\equiv \cad_{(a,\alpha)}(b,\beta):= \bra{\brac{a,b},\omega\bra{a,b}}\qquad a,b\in\alg\quad\alpha,\beta\in\Km,
\end{equation}
where $\cad$ is the associated adjoint action. It is readily seen
that \eqref{celb} is a well-defined Lie bracket if and only if
$\omega$ is a two-cocycle.

\begin{definition}
A two-cocycle on $\alg$ is a bilinear map
$\omega:\alg\times\alg\arrow\Km$ such that
\begin{enumerate}
\item[(i)] it is skew-symmetric: $\omega\bra{a,b} = -\omega\bra{b,a}$,
\item[(ii)] and it satisfies the cyclic condition:
\begin{equation}\label{jc}
\omega\bra{[a,b],c} + \omega\bra{[c,a],b} + \omega\bra{[b,c],a} = 0,
\end{equation}
\end{enumerate}
where $a,b,c\in\alg$.
\end{definition}

Notice that $(0,\alpha)\in\calg$ commute with respect to
\eqref{celb} with all other elements from $\calg$ and hence $\alg$
can be identified with $\alg\oplus\alpha$ for fixed $\alpha\in\Km$.
The value $\alpha$ is often called a charge. In fact
$\alg\cong\calg/\mathfrak{c}$, where $\mathfrak{c} =
\pobr{(0,\alpha)\in\calg : \alpha\in\Km}$ is in the center of
$\calg$, and thus the Lie algebra $\calg$ is called a {\it central
extension} of $\alg$.

Assume now (for simplicity) that $\alg$ can be identified with
$\alg^*$ through a non-degenerate symmetric scalar product
\eqref{sym}. Then this product can be extended to the algebra
$\calg$ in the following fashion:
\begin{equation}\label{csym}
    \bra{(a,\alpha),(b,\beta)}_{\calg} := \bra{a,b}_\alg + \alpha\beta\qquad a,b\in\alg\quad \alpha,\beta\in\Km.
\end{equation}
Of course, \eqref{csym} is symmetric and an important fact is that
it preserves non-degeneracy. Thus, $\calg^*$ can be identified with
$\calg$.

Usually two-cocycles are defined through the scalar product on
$\alg$ and a linear map $\phi:\alg\arrow\alg$ such that
\begin{equation}\label{2c}
    \omega\bra{a,b} = \bra{a, \phi(b)}_\alg.
\end{equation}
If the linear map $\phi:\alg\arrow\alg$ is skew-symmetric, i.e.,
$\phi^*=-\phi$, and the following condition
\begin{equation}\label{1c}
    \phi\bra{\brac{a,b}} = \ad^*_a\phi(b) - \ad^*_b\phi(a)
\end{equation}
is satisfied, then it is called a {\it one-cocycle}.

\begin{proposition}
The bilinear form given by \eqref{2c} is a two-cocycle if and only
if $\phi$ is a one-cocycle. Moreover, if the symmetric product on
$\alg$ is $ad$-invariant \eqref{adinv} then \eqref{2c} is a
two-cocycle if and only if the skew-symmetric $\phi$ is a derivation
of the Lie bracket $\brac{\cdot,\cdot}$ in $\alg$, i.e.,
\begin{equation}\label{dm}
    \phi\bra{\brac{a,b}} = \brac{\phi(a),b} + \brac{a,\phi(b)}
\end{equation}
holds.
\end{proposition}
\begin{proof}
It is clear that $\eqref{2c}$ is skew-symmetric if and only if
$\phi^*=-\phi$. The cyclic condition \eqref{jc} for skew-symmetric
$\phi$ has the form
\begin{align*}
 \omega\bra{\brac{a,b},c} + c.p. &= -\bra{\phi\bra{\brac{a,b}},c}_\alg + \bra{\brac{b,c},\phi(a)}_\alg + \bra{\brac{c,a},\phi(b)}_\alg\\
     &= -\bra{\phi\bra{\brac{a,b}},c}_\alg - \bra{\ad^*_b\phi(a),c}_\alg + \bra{\ad^*_a\phi(b),c}_\alg = 0,
\end{align*}
where we used the definition \eqref{coad} of coadjoint action. Now,
since the symmetric product on $\alg$ is non-degenerate, the cyclic
condition is equivalent to \eqref{1c}. For the $\ad$-invariant
symmetric product \eqref{dm} follows from \eqref{1c} since $\ad^*$
is identified with $\ad$.
\end{proof}

The adjoint action \eqref{celb} does not depend on $\alpha$, thus in
fact it defines adjoint action of $\alg$ on $\calg$. Hence, we can
omit dependence on the charge and write $\cad_a\equiv
\cad_{(a,\alpha)}$. When a given two-cocycle is defined by means of
a one-cocycle, i.e., \eqref{2c} holds, then we can write explicitly
the coadjoint action $\cad^*$, since
\begin{align*}
 \bra{\cad_b^*(a,\alpha),(c,\gamma)}_{\calg} &:=-\bra{(a,\alpha),\cad_b(c,\gamma)}_{\calg}
    = -\bra{a,\ad_bc}_\alg - \alpha \bra{b,\phi(c)}_\alg\\
     &= \bra{\ad_b^*a+\alpha\phi(b),c}_\alg = \bra{(\ad^*_ba+\alpha\phi(b),0),(c,\gamma)}_{\calg}.
\end{align*}
Hence, we can restrict $\cad^*$ to $\alg$ and define the central
extension of coadjoint action of $\alg^*$ on $\alg$, i.e., $\cad_b^* := \ad_b^* + \alpha \phi(b)$ for $b\in\alg$, where the charge $\alpha\in\Km$ is now treated as a parameter.

According to Theorem \ref{nlp} the Lie bracket \eqref{celb} on
$\calg$ defines a Lie-Poisson bracket on the space of smooth
functions on $\calg^*\cong\calg$, i.e., on $\smf(\calg)$. This
Poisson bracket can be restricted to $\smf(\alg)$ considered as a
subspace of $\smf(\calg)$. Hence, at a point $(L,\alpha)\in\calg$ we
have
\begin{equation}\label{celp}
\pobr{H,F}(L) := \bra{(L,\alpha),\brac{(dF,0),(dH,0)}}_{\calg}
= \bra{L,\brac{dF,dH}}_{\alg} + \alpha\,\omega\bra{dF,dH},
\end{equation}
where $H,F\in\smf(\alg)$ and $\alpha\in\Km$. The Poisson bracket
\eqref{celp} is a central extension of the natural Lie-Poisson
bracket on $\smf(\alg)$ generated by the Lie algebra structure on
$\alg$. When the two-cocycle is given by \eqref{2c}, then the
associated Poisson tensor $\pi$ such that $\pobr{H,F} = \bra{dF,\pi
dH}_\alg$ has the form $\pi dH = \cad^*_{dH} L \equiv \ad^*_{dH}L + \alpha \phi(dH)$.

The second Lie bracket on $\calg=\alg\oplus\Km$, being an extension
of \eqref{rbra}, is defined by
\begin{equation}\label{scelb}
    \brac{(a,\alpha),(b,\beta)}_R := \bra{\brac{a,b}_R,\omega_R\bra{a,b}}\qquad a,b\in\alg\quad\alpha,\beta\in\Km,
\end{equation}
where
\begin{equation}\label{r2c}
    \omega_R(a,b) = \omega(Ra,b) + \omega(a,Rb),
\end{equation}
and $\omega$ is a two-cocycle from \eqref{celb}. Clearly,
\eqref{scelb} is a well-defined Lie bracket on $\calg$ if and only
if $R:\alg\arrow\alg$ is a classical $R$-matrix and \eqref{r2c} is a
two-cocycle with respect to the second Lie bracket on $\alg$
\eqref{rbra} induced by $R$.

\begin{proposition}\label{scc}
A sufficient condition on $R$ for \eqref{r2c} to be a two-cocycle
with respect to \eqref{rbra} is the Yang-Baxter equation \eqref{YB}.
\end{proposition}
\begin{proof}
Skew-symmetry of \eqref{r2c} is obvious. Hence, it suffices to
verify the cyclic condition for the two-cocycles. Thus, from the
Yang-Baxter equation \eqref{YB} we have that
\begin{align*}
     &\omega_R\bra{\brac{a,b}_R,c} + c.p. = \omega\bra{R\brac{a,b}_R,c} + \omega\bra{\brac{a,b}_R,Rc} + c.p.\\
    &\qquad\quad = \omega\bra{\brac{Ra,Rb},c} + \alpha \omega\bra{\brac{a,b},c} + \omega\bra{\brac{Ra,b},Rc} + \omega\bra{\brac{a,Rb},Rc} + c.p.\\
     &\qquad\quad = \omega\bra{\brac{Ra,Rb},c} + \omega\bra{\brac{Rb,c},Ra} + \omega\bra{\brac{c,Ra},Rb} +c.p. = 0,
\end{align*}
where the last two equalities hold because $\omega$ is a
two-cocycle.
\end{proof}

Hence, the second Lie-Poisson bracket on $\smf(\alg)$ at a point
$(L,\alpha)\in\calg$ has the form
\begin{align}
\nonumber   \pobr{H,F}_R(L) &:= \bra{(L,\alpha),\brac{(dF,0),(dH,0)}_R}_{\calg}\\
\label{celp2}    &= \bra{L,\brac{dF,dH}_R}_{\alg} + \alpha\,\omega_R\bra{dF,dH},
\end{align}
where $H,F\in\smf(\alg)$ and $\alpha\in\Km$. When a two-cocycle is
given by \eqref{2c}, then the associated Poisson tensor $\pi_R$ such
that $\pobr{H,F}_R = \bra{dF,\pi_R dH}_\alg$ has the form
\begin{equation}\label{ept}
\pi_R dH = \cad^*_{RdH} L + R^*\cad^*_{dH} L \equiv \ad^*_{RdH}L + R^*\ad^*_{dH}L + \alpha \phi(RdH) + \alpha R^*\phi(dH).
\end{equation}
Notice that the higher-order Poisson tensors from Sections \ref{pa}
and \ref{na} do not survive the procedure of central extension. We
have the following straightforward extension of Theorem~\ref{rsts}.

\begin{theorem}
The Casimir functionals $C_n$ of \eqref{celp}, i.e.,
$C_n\in\smf(\alg)$ such that
\begin{equation}\label{t1}
    \pi dC_n = \ad^*_{dC_n}L + \alpha\, \phi(dC_n) = 0,
\end{equation}
are in involution with respect to \eqref{celp2} and hence generate
the following hierarchy of mutually commuting Hamiltonian evolution
equations on $\alg$:
\begin{equation}\label{t2}
    L_{t_n} = \pi_R dC_n = \ad^*_{RdC_n}L + \alpha\, \phi(RdC_n).
\end{equation}
\end{theorem}

Consider an important special case of \eqref{2c}, when $\phi$ is a
derivation with respect to additional continuous space coordinate.
Consider once more a Lie algebra $\alg$ with a non-degenerate
symmetric $\ad$-invariant scalar product \eqref{sc} defined by means of a
trace form $\tr$ on $\alg$. Assume now that $\alg$ depends in a
nontrivial fashion on an additional continuous parameter $y\in\Si$,
which naturally generates the corresponding current operator algebra
\begin{equation}\label{ca}
    \ca = \smf\bra{\Si,\alg}
\end{equation}
of smooth maps from $\Si$ to $\alg$. On the current algebra $\ca$ we
define the following modified trace form $\Tr:\ca\arrow\Km$, such that
$\Tr(a) := \int_{\Si} \tr(a)\,dy$, where $a\in\ca$.
Then the scalar product reads
\begin{equation}
(a,b)_{\ca}:= \Tr(ab) = \int_{\Si} \tr(ab)\,dy \qquad a,b\in\ca.
\end{equation}
Thus, assuming that the derivative with respect to $y$ is a
derivation of the Lie bracket in $\ca$, i.e., \eqref{dm} with $\phi
=\pr_y$ holds, we can define the so-called Maurer-Cartan two-cocycle
\begin{equation}\label{mc}
    \omega\bra{a,b} = \bra{a, \pr_y b}_{\ca}\equiv \int_{\Si} \tr(a \pr_yb)\, dy\qquad a,b\in\ca.
\end{equation}

In this case the Casimir functions of the natural Lie-Poisson
bracket on the centrally extended Lie algebra satisfy \eqref{t1} in
the form of the so-called Novikov-Lax equation
\begin{equation}\label{novlax}
    \brac{dC_n,L} + \alpha\, \pr_y(dC_n) = 0.
\end{equation}
This follows from the invariance of the scalar product, as in this
case $\ad^*\cong\ad$. The differentials of Casimirs $dC_n$  generate
the following Lax hierarchy \eqref{t2}
\begin{equation}\label{ceh}
    L_{t_n} = \brac{RdC_n,L} + \alpha\,\pr_y(RdC_n) = \pi_R dC_n,
\end{equation}
where the Poisson tensor \eqref{ept} takes the form
\begin{equation}\label{ept2}
    \pi_R dH = \brac{RdH,L} + R^*\brac{dH,L} + \alpha\,\pr_y(RdH) + \alpha\,R^*\pr_y(dH).
\end{equation}

\subsection{Dirac reduction and homotopy formula}\label{dir}

It often happens that we need to restrict the dynamics under study
to a submanifold defined via some constraints. In such a case, a
question arises of whether and how one can reduce the Poisson tensors.

Assume that the (linear) phase space $\U= \U_1\oplus \U_2$ is
spanned by $\uf_1\in \U_1$ and $\uf_2\in \U_2$, i.e. $\uf = (\uf_1,
\uf_2)^\T$. We will only consider the simplest case of the Dirac
reduction given by the constraint $\uf_2=\cf$, where $\cf\in \U_2$ is an
arbitrary constant. In many cases considering
such constraint is
sufficient. Besides, more complicated constraints can always be
reduced by change of dependent variables to several constraints of
the above type.

The Hamiltonian system with the Poisson tensor before reduction has
the form
\begin{equation*}
    \pmatrx{cc}{\uf_1\\ \uf_2}_t = \pmatrx{cc}{\pi_{11}(\uf_1, \uf_2) & \pi_{12}(\uf_1, \uf_2)\\
    \pi_{21}(\uf_1, \uf_2) & \pi_{22}(\uf_1, \uf_2)}\pmatrx{cc}{\var{H}{\uf_1}\\[1mm] \var{H}{\uf_2}},
\end{equation*}
where we assume that $\pi_{22}$ is nondegenerate and hence
invertible. Taking the constraint $\uf_2 = \cf$ into consideration
we find that
\begin{equation*}
    0 = \left .(\uf_2)_t \right |_{\uf_2=\cf} = \pi_{21}(\uf_1,\cf)\var{H}{\uf_1}
    + \pi_{22}(\uf_1,\cf)\var{H}{\uf_2}.
\end{equation*}
Thus, we can express $\var{H}{\uf_2}$ in the terms of
$\var{H}{\uf_1}$ and put it into
\begin{align*}
    (\uf_1)_{t} = \pi_{21}(\uf_1,\cf)\var{H}{\uf_1} + \pi_{22}(\uf_1,\cf)\var{H}{\uf_2}
    =: \pi^{red}(\uf_1)\var{H}{\uf_1}.
\end{align*}
Hence the reduced tensor takes the form
\begin{equation}\label{red}
  \pi^{red}(\uf_1)= \pi_{11}(\uf_1, \cf) -
  \pi_{12}(\uf_1, \cf)\cdot \brac{\pi_{22}(\uf_1, \cf)}^{-1}\cdot \pi_{21}(\uf_1, \cf).
\end{equation}

\begin{lemma}\cite{Dirac}\label{dirac}
The operator \eqref{red} is a Poisson operator on the affine space
$\U_1\oplus\cf$.
\end{lemma}

The skew-symmetry of \eqref{red} is obvious. However the proof of
the Jacobi identity for \eqref{red} consists of tedious yet
rather straightforward calculations \cite{Dirac} and we will omit
it. In the finite-dimensional case the said proof is much simpler
(see for example \cite{blmar}).

Notice that if the inner product is not $\ad$-invariant or we use
the central extension procedure, then in general we do not know the
explicit form of the Casimir functions (like \eqref{casq} for
example) for the natural Lie Poisson bracket (or the extension
thereof). In such a case one has to look for the annihilators $dC$
of the Lie-Poisson tensor by directly solving the  equation $\pi dC
=0$. With $dC$ in hand, one can try to reconstruct the Casimir
functions $C$. The Poincar\'e lemma says that if the phase space $\U$
is linear or of the star shape ($\forall u\in\U \pobr{\varepsilon u:
0\me\epsilon\me1}\subset\U$), then each closed $k$-form is exact. In
particular, when $\U$ satisfies the condition of the Poincar\'e lemma,
we can reconstruct the Casimir functions $C\in\smf{(\alg^*)}$ from
their differentials $dC\in\alg$ using the homotopy formula
\cite{Olver}
\begin{equation}\label{hom}
    C(\eta) = \int_0^1 \dual{dC(\epsilon\eta), \eta}d\epsilon\qquad \eta\in\alg^*.
\end{equation}
Nevertheless, even when \eqref{hom} is not applicable, the functions
$C$ can be often reconstructed through explicit integrations.

\subsection{Lax hierarchies from loop algebras}\label{la}

Let $\alg$ be a Lie algebra over $\Km$ with the Lie bracket
$\brac{\cdot,\cdot}$. In contrast with Section \ref{seclax}, we do
not assume here any additional structures on $\alg$.

\begin{definition}\label{loop}
The loop algebra over $\alg$ is the algebra $\alg^\la
:=\alg\bc{\la,\la^{-1}}$ of formal Laurent series in the parameter
$\la\in\Km$ with the coefficients from $\alg$.
\end{definition}

It is easily seen that thanks to bi-linearity the operation in the
former algebra extends to the loop algebra. Thus, in our case we can
readily extend the Lie bracket $\brac{\cdot,\cdot}$ to the loop
algebra $\alg^\la$ by setting
\begin{equation}\label{liel}
    \brac{a\la^m, b\la^n} = \brac{a,b}\la^{m+n}\qquad a,b\in\alg\quad m,n\in \Z.
\end{equation}

There are two natural decompositions of $\alg^\la$ into the sum of
Lie subalgebras, i.e.
\begin{equation}\label{ls}
    \alg^\la = \alg^\la_+\oplus\alg^\la_- = \bra{\sum_{i\me k}u_i\la^i}\oplus\bra{\sum_{i<k}u_i\la^i}
\end{equation}
for $k=0$ and $1$. Thus for $k=0$ and $k=1$ we have well-defined
classical $R$-matrices~\eqref{rp}
\begin{equation}\label{rm}
    R = P_+ -\frac{1}{2}.
\end{equation}
The transformation $\la\mapsto\la^{-1}$ maps the case of $k=0$ into
that of $k=1$, and vice versa. For this reason in what follows we
restrict ourselves to considering the case of $k=0$ only, and hence
$P_+$ and $P_-$ will stand for projections onto nonnegative and
negative powers of $\la$. In fact, we have an infinite family of classical $R$-matrices
\begin{equation}\label{rn}
    R_n = R \la^n\qquad n\in\Z
\end{equation}
and the corresponding new Lie brackets on $\alg$ read
\begin{equation}\label{rbr}
    \brac{a,b}_{R_n} := \brac{R_n a,b} + \brac{a,R_n b},\qquad a,b\in\alg.
\end{equation}
The $R$-matrices \eqref{rn} are well-defined since $\la^n$ is an
intertwining operator and Proposition~\ref{int} holds.

Let $L$ be an element of $\alg$. We have the following Lax hierarchy
\begin{equation}\label{lh}
     L_{t_n}= \brac{\bra{\la^nL}_+, L} = -\brac{\bra{\la^nL}_-, L}\qquad n\in\Z.
\end{equation}
The commutativity of the flows \eqref{lh} for different $n$ follows
from the fact that \eqref{rm} satisfy the Yang-Baxter equation
\eqref{YB} and $X_n(L) = \la^nL$ are invariants \eqref{inv} such
that $\bra{X_m(L)}_{t_n} = \brac{RX_m(L), X_n(L)}$ still holds. The details of computations
are parallel to those from the proof of Proposition~\ref{ybp}.

Without concentrating the specific properties of $\alg$, we can
investigate the general form of appropriate Lax operators from
$\alg^\la$, i.e., the operators $L\in\alg^\la$ that generate
self-consistent evolution equations on $\alg^\la$ from the Lax
equations \eqref{lh}. This means that the maximal and minimal orders
in $\la$ of right- and left-hand sides of \eqref{lh} have to
coincide. Consider a bounded Lax operator $L\in\alg^\la$ of the form
\begin{equation}\label{bl}
  L = u_N\la^N + u_{N-1}\la^{N-1} + \ldots + u_{1-m}\la^{1-m} + u_{-m}\la^{-m}\qquad N,m\in\Z,
\end{equation}
where $N\me -m$ and $u_i\in\alg$. Then a straightforward analysis
shows that \eqref{bl} yields consistent equations \eqref{lh} if
$u_N$ is a nonzero time-independent element of $\alg$, i.e.,
$(u_N)_{t_n}=0$. The specific properties of $\alg$ might lead to
further restrictions on \eqref{bl}.

Theorem \ref{rsts} can be applied for the construction of
Hamiltonian hierarchies on the dual algebra to $\alg^\la$. We
understand the dual algebra as $\alg^{*\la} =
\alg^*\bc{\la,\la^{-1}}$. However, in contrast with the previously
considered algebras, in the case of a loop algebra we have an
infinite family of $R$-matrices \eqref{rn} with respective
Lie-Poisson brackets \eqref{rliepo} and related Poisson tensors of
the form $\pi_{R_l} dH = \ad^*_{R_l dH}\eta + R_l^*\ad^*_{dH}\eta$, where $\eta\in\alg^{*\la}$
and $dH\in\alg^\la$.

The coadjoint action $\ad^*$ is defined with respect to the Lie
bracket \eqref{liel} on $\alg^\la$. All the above Lie-Poisson
brackets are mutually compatible, which follows from the fact that
the sum of intertwining operators also is an intertwining operator.
Besides, if $C_n\in\smf\bra{\alg^{*\la}}$ is a Casimir function of
the natural Lie-Poisson bracket \eqref{nliepo}, i.e. $\ad^*_{dC_n}\eta = 0$,
then $\la^l C_n$ is also a Casimir function. Hence, in the loop
algebra case the Casimir functions generate multi-Hamiltonian
hierarchies of mutually commuting vector fields on $\alg^{*\la}$
\begin{equation*}
  \eta_{t_n} = \ad^*_{R dC_n}\eta = \ldots = \pi_{R_{-1}}dC_{n+1} = \pi_{R} dC_{n}
  = \pi_{R_{1}}dC_{n-1} = \ldots,
\end{equation*}
where $dC_{n+l} = \la^l dC_{n}$ for $l\in\Z$.

If we have a trace form on $\alg$, given by a linear map
$\tr:\alg\arrow\Km$ such that the form in question is non-degenerate
and symmetric, then this form can be extended to the loop algebra
$\alg^\la$ by the formula
\begin{equation}\label{ltr}
    \Tr (a) = \tr (\res\ a)\qquad a\in\alg^\la,
\end{equation}
where $\res \sum_i a_i\la^i = a_{-1}$. In fact, the choice of
residue in \eqref{ltr} is a matter of convention, and one can choose
the coefficient of an arbitrary order to get a proper definition of
a trace. This is in contrast with the trace form \eqref{dtr} in the
case of Poisson algebras. Nondegeneracy and symmetry of
\eqref{ltr} are preserved. Moreover, if $\tr$ defines an
$\ad$-invariant metric on $\alg$, then this is also true for
\eqref{ltr}, and one can make an identification $\ad^*\equiv \ad$.

Consider now the central extension approach for the loop algebras.
Assume for simplicity that on $\alg$, and therefore on $\alg^\la$,
we have a nondegenerate inner product, and hence
$\alg^{*\la}\cong\alg^\la$. Let $\calg^\la$ be an extension of
$\alg^\la$ with the Lie bracket \eqref{celb} defined by a
two-cocycle $\omega$. The extended natural Lie-Poisson bracket has
the form
\begin{equation}\label{enlp}
   \pobr{H,F}(L) := \bra{L,\brac{dF,dH}}_{\alg^\la} + \alpha\ \omega\bra{dF,dH}.
\end{equation}
In contrast with the previous case, we have here an infinite family
of $R$-matrices \eqref{rn} inducing an associated infinite family of
new Lie-Poisson brackets on $\smf\bra{\alg^\la}$,
\begin{align}\label{pbrn}
      \pobr{H,F}_n(L) :=  \bra{L,\brac{dF,dH}_{R_n}}_{\alg} + \alpha\ \omega_{R_n}\bra{dF,dH},
\end{align}
where
\begin{equation}\label{cocr}
    \omega_{R_n}\bra{a,b}:= \omega\bra{R_na,b} + \omega\bra{a,R_nb}\qquad a,b\in\alg^\la.
\end{equation}
All quantities \eqref{cocr} are two-cocycles of the respective
Lie brackets \eqref{rbr}. For $n=0$ according to Proposition \ref{scc}
this follows from the fact that $R$ satisfies the classical
Yang-Baxter equation, in the remaining cases one must additionally
use the fact that $\la^n$ are intertwining operators. It is
important to stress that all Poisson brackets \eqref{pbrn} are
pairwise compatible, which follows from the fact that the linear sum
of intertwining operators is an intertwining operator. If the two-cocycle $\omega$ is given in the
form \eqref{2c}, then the Casimirs of the extended natural
Lie-Poisson bracket \eqref{enlp} satisfy
\begin{equation}\label{dc}
    \pi dC_n = \ad^*_{dC_n}L + \alpha\, \phi(dC_n) = 0,
\end{equation}
and generate multi-Hamiltonian Lax hierarchy
\begin{equation}\label{llh}
 L_{t_n} = \ad^*_{R dC_n}L + \alpha\, \phi(RdC_n) = \ldots = \pi_{R_{-1}}dC_{n+1} = \pi_{R} dC_{n}
  = \pi_{R_{1}}dC_{n-1} = \ldots,
\end{equation}
where $L\in\alg^\la$, with the restriction that $dC_{n+l} = \la^n
dC_{n}$ for $l\in\Z$. The respective Poisson tensors associated with
the brackets \eqref{pbrn} are given by
\begin{equation}\label{pts}
    \pi_{R_l} dH = \ad^*_{R_ldH}L + R_l^*\ad^*_{dH}L + \alpha\, \phi(R_ldH) + \alpha\, R_l^*\phi(dH).
\end{equation}
We readily see that $R_l^\star =-\la^lR$. To find $dC_n$, we can
assume that
\begin{equation*}
    dC_n\equiv \la^n dC_0 = a_0 + a_{1}\la^{1} + a_{2}\la^{-2} + \ldots\qquad n\in\Z,
\end{equation*}
and solve \eqref{dc} recursively for the coefficients $a_i$. Then
the functions $C_n$ can be reconstructed using the homotopy formula
\eqref{hom}.

Consider now a particular case of the Maurer-Cartan two-cocycle
\eqref{2c}, with $\phi = \pr_x$ and the loop algebra $\ca^\la$ over
\eqref{ca}. Assuming that $\ad^\star\equiv\ad$, the
Lax hierarchy \eqref{llh} takes the form
\begin{equation}\label{lh1}
  L_{t_n} = \brac{\bra{dC_n}_+,L} + \alpha\,\pr_x(dC_n)_+ = \ldots = \pi_{R_l}dC_{n-l} = \ldots\ .
\end{equation}
Then, analyzing \eqref{lh1} we find that $L\in\ca^\la$ of the form
\eqref{bl} yields self-consistent equations for $\alpha\neq 0$ if
$N\me -1$ and $u_N$ is nonzero and time-independent (except for
$N=-1$) and $m\me 0$. The Hamiltonian structures for \eqref{lh1} are
given by the Poisson tensors \eqref{pts} taking the form
\begin{equation}\label{pt1}
   \pi_{R_l} dH = \brac{\bra{\la^ldH}_+,L} - \la^l\brac{dH,L}_+ + \alpha\, \pr_x(\la^ldH)_+ - \alpha\,\la^l\pr_x(dH)_+.
\end{equation}
The Poisson tensors \eqref{pt1} form a proper subspace of $\ca^\la$ with respect
to \eqref{bl}, with the above restrictions, if $N\me l\me -m$ for
$N\me 0$ and if $0\me l\me -m$ for $N=-1$. Thus, there always exist at
least two Poisson tensors $\pi_{R_l}$ for which the procedure of
Dirac reduction is not required. Recall that this analysis
disregards specific properties of the Lie algebra $\alg$. Thus, if
$L$ is further constrained according to these specific properties
the Dirac reduction might be yet required.\looseness=-1

Finally, let us note that in contrast with the previously considered
algebras, where the central extension has lead to (2+1)-dimensional
systems, in the case of loop algebras the central extension is
necessary for the construction of (1+1)-dimensional integrable
continuous field systems. The reader will find the examples of this
construction in the next sections.

A natural choice are loop algebras defined over finite-dimensional
semi-simple Lie algebras. In such a case the Killing form gives us
symmetric, nondegenerate and also $\ad$-invariant inner product. Thus, taking into consideration the central extension
procedure with the Maurer-Cartan two-cocycle, the above choice leads to the constructions of
a wide class of (1+1)-dimensional integrable continuous systems. The simplest case is
presented in Section \ref{sl2}.

\section{Integrable dispersionless systems}

The theory of integrable dispersionless or equivalently integrable
hydrodynamic-type systems, i.e., the  quasi-linear systems of
first-order partial differential equations, belongs to the most
recent ones and has been systematically developed from the 1980's.
Significant progress was achieved after S.~Tsarev \cite{Tsarev}
discovered a technique called the generalized hodograph method that
permits to find solutions using quadratures, see also \cite{fer1}.

The study of the Poisson structures of dispersionless systems was
initiated by B.~Dubrovin and S.~Novikov \cite{Dubrovin1}. They
established a remarkable result that the Poisson tensors of hydrodynamic
type can be generated by contravariant nondegenerate flat Riemannian
metrics. The natural geometric setting for
the associated bi-Hamiltonian structures (Poisson pencils) is the
theory of Frobenius manifolds based on the geometry of pencils of
contravariant Riemannian metrics \cite{dubr}. The Frobenius
manifolds were introduced by B.~Dubrovin as a
coordinate-free form of the associativity equations, appearing in the
context of deformations of $2$-dimensional topological field theories
(TFT), the so-called WDVV equations, that can be identified with
(a class of) hydrodynamic-type systems.

\subsection{Poisson algebras of Laurent series}

Consider the algebra of 'formal' Laurent series in $p\in\Cm^*$ about
$\infty$,
\begin{align*}
\Alg = \Alg_{\me k-r} \oplus \Alg_{<k-r} := \pobr{\sum_{i= k-r}^N
a_i(x)p^i}\oplus \pobr{\sum_{i< k-r} a_i(x)p^i},
\end{align*}
where $u_i$ are smooth functions of continuous variable
$x\in\Omega$, i.e., $\Alg$ consists of polynomial functions in $p$
and $p^{-1}$ with finite highest orders, where $\Omega = \Si$ if we
assume these functions to be periodic in $x$ or $\Omega = \Rm$ if
these functions belong to the Schwartz space ($u_i$ and all their
derivatives tend rapidly to zero when $x$ approaches $\pm\infty$).
We can introduce the Lie algebra structure on $\Alg$ in infinitely
many ways using a family of Poisson brackets
\begin{equation}\label{pb}
\pobr{f, g}_r := p^r
\bra{\pd{f}{p}\pd{g}{x}-\pd{f}{x}\pd{g}{p}}\qquad r\in \bb{Z}\quad
f,g\in \Alg.
\end{equation}
that generalize the well-known canonical Poisson bracket (the case
$r=0$).

The trace form in the algebra $\Alg$ with fixed Poisson bracket
\eqref{pb} for some $r$ is defined in the following fashion:
\begin{align}\label{dtr}
\Tr f := -\int_\Omega \res_\infty \bra{p^{-r}f} dx\qquad f\in\Alg.
\end{align}
Here $\res$ is the standard residue at $p=\infty$ such that
$\res_\infty L = -u_{-1}$ for $L=\sum_i u_ip^i$.

\begin{proposition}
The scalar product \eqref{sc} defined by means of the trace form
\eqref{dtr} is symmetric, nondegenerate and $\ad$-invariant.
\end{proposition}
\begin{proof}
The nondegeneracy and symmetry are obvious. Let $\gamma$ be a closed
curve encircling once an infinity point on the extended complex
plane. Then
\begin{align*}
 \Tr \pobr{f,g}_r &= -\int_\Omega \res_\infty \bra{\pr_p f \pr_x  g}
dx + \int_\Omega \res_\infty \bra{\pr_x  f \pr_p g} dx\\
 &= -\frac{1}{2\pi i}\int_\Omega \oint_{\gamma_\la} \bra{\pr_p \pr_x
g} dpdx + \frac{1}{2\pi i}\int_\Omega \oint_{\gamma_\la} \bra{\pr_x
f \pr_p g} dpdx = 0,
\end{align*}
where the latter equality follows from integrations by parts with
respect to $p$ and $x$. Hence, the $\ad$-invariance follows by Lemma
\eqref{adtr}.
\end{proof}

Obviously, \eqref{pb} is a derivation for the multiplication in
$\Alg$ and hence we can further apply the scheme from Section 2.5.
Fixing $r$ we fix Poisson algebra and we are able to construct
$R$-matrices following from the decomposition of $\Alg$. Simple
inspection shows that $\Alg_{\me k-r}$ and $\Alg_{< k-r}$ are Lie
subalgebras of $\Alg$ only in the following cases:
\begin{itemize}
\item if $r=0$ for $k=0$;
\item if $r\in\Z$ for $k=1,2$;
\item if $r=2$ for $k=3$.
\end{itemize}
Thus, fixing $r$ we fix the Lie algebra structure with $k$ numbering
the $R$-matrices given in the form $R=P_{\me k-r} - \frac{1}{2}$.
Hence, we have multi-Hamiltonian Lax hierarchies \eqref{eveq} of the
form
\begin{equation}\label{dlaxh}
    L_{t_n} = \pobr{\bra{L^{\frac{n}{N}}}_{\me k-r},L}_r = \pi_0 dH_n = \pi_1 dH_{n-1}\qquad n=1,2,\ldots,
\end{equation}
generated by fractional powers of infinite-field Lax functions
$L\in\Alg$ given in the form
\begin{equation}\label{pol}
L =  u_N p^N + u_{N-1}p^{N-1} + u_{N-2} p^{N-2} + \ldots\qquad N\neq
0,
\end{equation}
where $u_{N}=1, u_{N-1}=0$ for $k=0$ and $u_{N}=1$ for $k=1$. A
simple analysis of \eqref{dlaxh} shows that \eqref{pol} are
appropriate Lax functions except for the case $k=3$, which is
excluded from further  considerations. This means that all the Lax
functions under study have pole or zero at infinity.

We find that $R^* = \frac{1}{2} - P_{\me 2r-k}$. The differentials
of a given functional $H\in\smf\bra{\Alg}$ of \eqref{pol} have the
form $dH = \sum_{i=-\infty}^{N+k-2}\var{H}{u_i}p^{r-1-i}$. Hence,
the related Poisson tensors for \eqref{dlaxh} are given by
\eqref{pot}
\begin{equation}\label{poi}
    \pi_q dH  = \pobr{\bra{L^q dH}_{\me k-r},L}_r - L^q\bra{\pobr{dH,L}_r}_{\me 2r-k},
\end{equation}
where $q=0,1,\ldots$, with the following Hamiltonians \eqref{casq}
\begin{equation}\label{dham}
    H_n = -\frac{N}{n+N} \Tr \bra{L^{\frac{n}{N}+1}}\qquad n\neq -N.
\end{equation}

We still have to check whether the above Lax functions span proper
subspaces with respect to the above Poisson operators \eqref{poi},
of the full Poisson algebras. We will restrict our considerations to
linear ($n=0$) and quadratic ($n=1$) Poisson tensors, as they are
obvious enough to define bi-Hamiltonian structures. Besides, in all
nontrivial cases, the Lax functions do not span proper subspaces
w.r.t. the Poisson tensors for $n\me 2$.

For $k=0$ the above Lax functions always span the proper subspace
w.r.t. the linear Poisson tensor, but for $k=1,2$ this is the case
only if $N\me 2r-2k+1$, otherwise the Dirac reduction is required.
For the quadratic Poisson tensors the Dirac reduction is always
necessary. The reduced quadratic Poisson tensor for $k=r=0,1,2$ is
given by
\begin{equation}\label{quad_I}
 \pi_1^{red}dH  = \pobr{\bra{L dH}_{\me 0},L}_r - L
\bra{\pobr{dH,L}_r}_{\me r} +
\frac{1}{N}\pobr{L,\pr_x^{-1}\res_\infty \pobr{dH,L}_0}_r,
\end{equation}
see Section \ref{dir} and Lemma \eqref{dirac}. The reduced Poisson
tensors \eqref{quad_I} are always local as the residue from the last
term is always a total derivative.

The dispersionless systems with the Lax representations of the form
\eqref{dlaxh} where $r=0$ with $k=0,1$ (the canonical Poisson
bracket) related to the dispersionless KP hierarchy and the dispersionless
modified KP hierarchy and the case $k=r=1$ of the dispersionless Toda
hierarchy, together with their finite-field reductions, were considered
in many papers (see for example
\cite{leb1,zakh,kod1,kri,Takasaki2,Strachan1,Aoyama2,Fairlie,bru}
and more recent \cite{ch,CDZ,Pa1}). The theory with the Poisson bracket
\eqref{pb} for arbitrary integer $r$, from the point of view of
classical $R$-matrices, was considered for the first time in
\cite{Blaszak5} and further developed in
\cite{Blaszak6,Szablikowski}.

The decomposition of the algebra $\Alg$ into Lie subalgebras is
preserved under the transformation $p\map p^{-1}$ (the case $k,r$
goes to $k'=3-k,r'=2-r$), but the Laurent series at $\infty$
\eqref{pol} transform into Laurent series at zero. This fact
suggests more analytic approach to the construction of
dispersionless systems, see \cite{Krichever2,GMMA}. For the theory
of meromorphic Lax representations \eqref{dlaxh} of dispersionless
systems see \cite{Szablikowski}. In \cite{Krichever2} Krichever
introduced the so-called universal Whitham hierarchies by means of
the moduli spaces of Riemann surfaces of all genera. Other classes
of reductions yielding nonstandard integrable dispersionless systems
can be found for instance in \cite{Gibbons,Yu,Pa3}. For a more
detailed discussion of the Lax representations \eqref{dlaxh} for
noncanonical Poisson brackets, their reductions and Hamiltonian
structures as well as several examples of finite-dimensional
reductions see \cite{Blaszak6} and \cite{Szablikowski}.

We know that the quantization of Poisson algebras for $r=0$ (and for
the equivalent case $r=2$) gives the algebra of pseudo-differential
operators and leads to the construction of field soliton systems,
while the quantization of the case for $r=1$ gives the algebra of
shift operators and leads to the lattice soliton systems. However,
the class of reductions yielding construction of dispersionless
systems is much wider then the class of corresponding reductions
yielding systems with dispersion. Besides, the issue of quantization
of Poisson algebras for $r\neq 0,1,2$ that would lead to the construction of
dispersive integrable systems is still open, see
\cite{Strachan2,Blaszak8,Das,szabla} and references therein.

Below we present two examples of hydrodynamic chains
(infinite-field systems) with their bi-Hamiltonian structures. For
the classifications of hydrodynamical chains and related
hydrodynamic Poisson tensors see \cite{Pa0,Pa1,Pa2}.

\begin{example}
Let us consider the infinite-field Lax function \eqref{pol} for $N=1$:
 $L= p + \sum_{i\me 0} u_i p^{-i-1}$. Then we find the first nontrivial
hydrodynamic chain from hierarchy \eqref{dlaxh}, the well-known
Benney moment chain,
\begin{align*}
  L_{t_2} = \frac{1}{2}\pobr{\bra{L^2}_{\me 0},L}_0 = \pi_0 dH_2 = \pi_1 dH_1\iff
(u_i)_{t_2} = (u_{i+1})_x+iu_{i-1}(u_0)_x,
\end{align*}
where $(L^2)_{\me 0} = p^2 + 2u_0$. The bi-Hamiltonian structure is
given by the Poisson tensors
\begin{align*}
\pi_0^{ij} &= j\pr_xu_{i+j-1}+iu_{i+j-1}\pr_x\\
\pi_1^{ij} &=(i+1)u_{i+j}\pr_x+(j+1)\pr_x u_{i+j}+(i+1)ju_{i-1}\pr_x u_{j-1}\\
&\qquad +\sum_{k=0}^{j-1}\brac{(i-j+k)u_{i-1+k}\pr_xu_{j-1-k}+k\, u_{j-1-k}\pr_x
u_{i-1+k}},
\end{align*}
which are obtained from \eqref{poi} for $q=0$ and \eqref{quad_I}.
The respective Hamiltonians are \eqref{dham}
\begin{equation*}
    H_1 = \frac{1}{2}\int_{\Omega} u_1\,dx\qquad H_2 = \frac{1}{2}\int_{\Omega}\bra{u_0^2+u_2} dx.
\end{equation*}
\end{example}

\begin{example}
The case of $k=r=1$ with the Lax function \eqref{pol} for $N=1$ of the
form $L = p + \sum_{i=0}^\infty u_{i}p^{-i}$. The first
hydrodynamic chain from \eqref{dlaxh} has the following form:
\begin{equation*}
L_{t_1} = \pobr{(L)_{\me 0},L}_1=\pi_0 dH_1 = \pi_1 dH_0\iff
\label{e1} (u_i)_{t_1} = (u_{i+1})_x+iu_i(u_0)_x.
\end{equation*}
The bi-Hamiltonian structure is given by the Poisson tensors
\begin{align*}
\pi_0^{ij} &= j\pr_xu_{i+j}+iu_{i+j}\pr_x\\
\pi_1^{ij} &=
\sum_{k=0}^i\brac{(j-k)u_k\pr_xu_{i+j-k}+(i-k)u_{i+j-k}\pr_xu_k}
+ i(j+1)u_i\pr_xu_j\\
&\qquad + (j+1)\pr_xu_{i+j+1}+(i+1)u_{i+j+1}\pr_x
\end{align*}
as well as the Hamiltonians $H_0 = \int_{\Omega} u_0\,dx$
and $H_1 = \int_{\Omega} \bra{u_1+\frac{1}{2}u_0^2} dx$.
\end{example}

The central extension approach from Section \ref{cea} with the
Mauren-Cartan two-cocycle yields the following Lax hierarchy
\eqref{ceh} of (2+1)-dimensional hydrodynamic systems
\begin{equation}\label{dlaxh2}
    L_{t_n} = \pobr{\bra{dC_n}_{\me k-r},L-\alpha\,q}_r = \pi_0 dH_n\qquad n=1,2,\ldots,
\end{equation}
where the Poisson bracket takes the form
\begin{equation*}
 \pobr{f, g}_r := p^r \bra{\pd{f}{p}\pd{g}{x}-\pd{f}{x}\pd{g}{x}} +
\bra{\pd{f}{q}\pd{g}{y}-\pd{f}{y}\pd{g}{q}}\quad r\in \bb{Z}\quad
f,g\in \widetilde{\Alg}.
\end{equation*}
The infinite-field Lax operators \eqref{pol}, with $u_i$ depending additionally on $y\in\Si$, are admissible with respect to \eqref{dlaxh2} if $N\me 1-r$. To construct evolution
equations from \eqref{dlaxh2}, we take $dC_n = \sum_{i=0}^\infty
a_{n-i}p^{n-i}$ and solve the Novikov-Lax equation \eqref{novlax},
$\pobr{dC_n, L-\alpha\,q}_r = 0$ for the auxiliary fields $a_i$ in
terms of the fields from Lax function. The linear Poisson tensor
\eqref{ept} is given by
\begin{equation*}
    \pi_0 dH  = \pobr{\bra{dH}_{\me k-r},L-\alpha\,q}_r - \bra{\pobr{dH,L-\alpha\,q}_r}_{\me 2r-k}.
\end{equation*}
The Lax functions \eqref{pol} span the proper subspace w.r.t. the
above linear Poisson tensor if $N\me 2r-2k+1$, otherwise the Dirac
reduction is required. The construction of (2+1)-dimensional
dispersionless systems by means of central extension procedure
yielding the Lax hierarchy was presented in \cite{Blaszak7}, where
one can also find a number of examples.

\begin{example}
The case of $k=r=0$. Consider $L =p^2 + u$. Then, for $(dC_3)_{\me 0} = p^3 + \frac{3}{2}up + \frac{3}{4}\alpha \pr_x^{-1}u_y$ we obtain the (2+1)-dimensional dKP equation
\begin{equation*}
 L_{t_3} = \pobr{(dC_3)_{\me 0},L-\alpha\,q}_0=\pi_0 dH_3 = \pi_1
dH_1\iff u_{t_3} = \frac{3}{2}uu_x +
\frac{3}{4}\alpha^2\pr_x^{-1}u_y,
\end{equation*}
where the Poisson tensors are $\pi_0 = 2\pr_x$ and if $\alpha=0$
$\pi_1 = \pr_xu+u\pr_x$. The Hamiltonians are
\begin{equation*}
H_1 = \iint_{\Omega \times \Si}\frac{1}{4}u^2\ dxdy\qquad H_3 =
\iint_{\Omega \times
\Si}\frac{1}{16}\bra{2u^3+3\alpha^2u\pr_x^{-2}u_{yy}}dxdy.
\end{equation*}
\end{example}

\begin{example}
Consider the Lax operator $L = p^{2-r} + up^{1-r} + vp^{-r}$ for
$k=1$ and $r\in\Z$, $r\neq 2$. Then for $(dC_{2-r})_{\me -r+1} =
p^{2-r} + up^{1-r}$ we have
\begin{align*}
L_{t_{2-r}} &= \pobr{(dC_{2-r})_{\me 1-r},L-\alpha\,q}_r=\pi_0 dH_{2-r} = \pi_1 dH_{1-r}\iff\\
&\qquad\pmatrx{cc}{u\\ v}_{t_{2-r}} = \pmatrx{cc}{(2-r)v_x+\alpha u_y\\
ru_xv+(1-r)uv_x}.
\end{align*}
For $r=1$ we get the (2+1)-dimensional dispersionless Toda equation
with Hamiltonian structure given by
\begin{equation*}
\pi_0 = \pmatrx{cc}{\alpha\pr_y & \pr_xv\\
v\pr_x & 0 }\quad{\rm and,\,if\,\alpha=0},\quad\pi_1 = \pmatrx{cc}{\pr_xv + v\pr_x & u\pr_x v\\
v\pr_x u & 2v\pr_x v},
\end{equation*}
with the Hamiltonians
\begin{equation*}
H_1 = \iint_{\Omega\times\Si}\bra{v+\frac{1}{2}u^2}dxdy\qquad
H_0=\iint_{\Omega\times\Si}u\,dxdy.
\end{equation*}
\end{example}

\subsection{Universal hierarchy}

Let $\alg = \Vect(\Si)$ be the Lie algebra of (smooth) vector fields
on the circle $\Si$ over the field $\Km$. The elements of
$\Vect(\Si)$ can be identified with smooth functions $a(x)$ of
spatial variable $x\in\Si$, with a Lie bracket in $\alg$ of the form
\begin{equation}\label{comvir}
 \dual{a,b} := ab_x-b a_x,
\end{equation}
where $a,b\in\alg$. Notice that \eqref{comvir} is a well-defined Lie
bracket that does not satisfy the Lebniz rule \eqref{derm}. Thus, we
follow the scheme from Section \ref{la} and consider the loop
algebra over $\Vect(\Si)$, i.e., $\alg^\la =
\Vect(\Si)\bc{\la,\la^{-1}}$. The commutator \eqref{comvir} readily
extends to $\alg^\la$. We already know that we have the decomposition
of $\alg^\la$ into Lie subalgebras \eqref{ls}, with respective
classical $R$-matrices generating the following Lax hierarchies
\eqref{lh}:
\begin{equation}\label{uh}
    L_{t_n} = \dual{\bra{\la^n L}_{\me k}, L}\qquad k=0,1,
\end{equation}
where $L\in\alg^\la$ and $n\in\Z$.

The hierarchy \eqref{uh} is the so-called universal hierarchy of
hydrodynamic type, which is a subject of intensive research in
recent years \cite{as0,as1}. The hierarchy \eqref{uh} can be
obtained as a quasi-classical limit of the coupled KdV equations of
Antonowicz and Fordy \cite{ant}. In this fashion we can obtain
multi-Hamiltonian structure for the universal hierarchy \cite{fer2}.
In fact, the multi-Hamiltonian structure of the coupled KdV
equations using classical $R$-matrix approach can be algebraically
interpreted as a set of compatible Lie-Poisson structures on the
dual space to the loop Virasoro algebra \cite{for}. The Virasoro
algebra is a central extension of the Lie algebra of vector fields
on a circle $\Vect(\Si)$ associated to the Gelfand-Fuchs two-cocycle
(i.e., \eqref{2c} with $\phi = \pr_x^3$).

\begin{example}
Consider \eqref{uh} for the infinite-field Lax functions given in
the following (appropriate) form $L = u_0 + u_1\la^{-1} + u_2\la^{-2} + \ldots$,
where $u_0=1$ for $k=0$. One can observe that $L_{t_n} =
\dual{\bra{\la^n L}_{\me k}, \bra{\la^nL}_{<k}}\la^{-n}$. Hence, the
evolution equations from \eqref{uh} take the form of
(1+1)-dimensional hydrodynamic chains
\begin{align*}
    (u_i)_{t_n} = (1-k)(u_i)_x + \sum_{j=1-k}^{i-1+k}\dual{u_{i-j},u_{n+j}}\qquad k=0,1.
\end{align*}
\end{example}

Assume that the dynamical fields in $\alg^\la$ depends on an
additional spatial variable~$y$. Let us now consider a
(2+1)-dimensional counterpart of \eqref{uh}. It reads
\begin{equation}\label{uh2}
    L_{t_n} = \dual{\bra{A_n}_{\me k}, L} + \pr_y\bra{A_n}_{\me k}\qquad k=0,1,
\end{equation}
where $A_n = a_n\la^n + a_{n-1}\la^{n-1} + \ldots$ satisfy
\begin{equation}\label{ln}
    \dual{A_n,L} + \bra{A_n}_y = 0.
\end{equation}
Notice that $A_n = \la^n A_0$. For a given $L\in\alg^\la$ one finds
coefficients $a_i$ from $A_n$ by solving \eqref{ln} recursively.
Notice that one cannot obtain \eqref{uh2} as a central extension of
the universal hierarchy, as \eqref{mc} is not a two-cocycle
associated with $\Vect(\Si)$. Commutativity of the equations from
the hierarchy \eqref{uh2} can be proved by straightforward
computation. On the other hand, integrability of the equations from
\eqref{uh2} follows from the fact that \eqref{uh2} is a Lax
sub-hierarchy of the centrally extended cotangent universal
hierarchy \cite{serg2}.

\begin{example}
Let $L$ have the form $L = \la + u$ in the case of $k=0$. Then
solving \eqref{ln} one finds that $A_0 = 1 + u\la^{-1} + \pr_x^{-1}u_y\la^{-2} + \ldots$.
Hence, the first nontrivial (2+1)-dimensional hydrodynamic equation from the
hierarchy \eqref{uh2} is \cite{as1}
\begin{equation*}
    L_{t_2} = \dual{\bra{A_2}_{\me 0},L} + \pr_y\bra{A_n}_{\me 0}\iff u_{t_2} = \pr_x^{-1}u_{yy} - uu_y + u_x\pr_x^{-1}u_y.
\end{equation*}
This hydrodynamic equation is equivalent to system of the form
\begin{equation*}
    u_{t_2} - v_y + uv_x - u_xv =0\qquad v_x-u_y =0,
\end{equation*}
that has recently attracted considerable attention, see \cite{Pa0,fer3,dun,man,or,serg2}.
\end{example}

\begin{example}
The case of $k=1$ for $L = u\la^{-1}$. We have $A_0 = 1 -
\pr_y^{-1}u_x\la^{-1} + \ldots$. Hence
\begin{equation*}
    L_{t_2} = \dual{\bra{A_2}_{\me 1},L} + \pr_y\bra{A_n}_{\me 1}\iff u_{t_2} = u\pr_y^{-1}u_{xx} - u_x\pr_y^{-1}u_x.
\end{equation*}
\end{example}

\section{Integrable dispersive systems}

In the following section we apply the general formalism to several infinite dimensional Lie algebras in order to construct a vast family of integrable dispersive (continuous and discrete soliton) systems.

Recently, the so-called integrable $q$-analogues of KP and Toda
like hierarchies become of increasing interest, see \cite{kass,fre,klr,ahm,tak}.
Our approach presented in two following subsections includes the $q$-systems as a special case.
Actually, we consider generalized algebras of shift and pseudo-differential operators
that allows to construct in one scheme not only ordinary lattice and field systems, but
in particular also their $q$-deformations.

We also present $(2+1)$-dimensional extensions of the Lax hierarchies following
from the algebras of shift and pseudo-differential operators.
In these cases the the quadratic Poisson tensor is not preserved.
Nevertheless, the second Poisson tensor can be given by means of the
so-called operand formalism \cite{san,fok,magn}. Its construction, within classical $R$-matrix
approach, can be found in \cite{Blaszak3}.

In last subsection we present the application of classical $R$-matrix formalism
to the loop algebra  $sl(2,\Cm)\bc{\la,\la^{-1}}$, which is the simplest
case of infinite-dimensional Lie algebras of Kac-Moody. In fact the scheme
of construction of infinite-dimensional systems from affine Kac-Moody Lie algebras is one
of the most general and particular cases are closely related
to the algebras of shift and pseudo-differential operators yielding soliton systems.
For details we send the reader to the survey \cite{drinfeld}.

For the application of $R$-matrix formalism to other algebras  see for example \cite{Pr1,Pr2,Pr3,skryp2}. A specially interesting class of algebras is related with the so called super-symmetric (SUSY) systems. The details the reader finds in \cite{O,P1,bru1,Pr2,He} and in literature quoted there.

\subsection{Algebra of shift operators}

Consider the algebra of `formal' shift operators
\begin{equation}\label{aso}
  \alg = \alg_{\me k-1} \oplus \alg_{< k-1} = \pobr{\sum^N_{i\me k-1}
u_i\e^i} \oplus \pobr{\sum_{i< k-1} u_i\e^i},
\end{equation}
where $u_i\in\F$, and $\F$ is an algebra of dynamical fields with
values in $\Km$. The associative multiplication rule in $\alg$
\eqref{aso} is defined by
\begin{equation}\label{smult}
\e^m u = E^m (u)\e^m\qquad m\in \Z.
\end{equation}

\begin{proposition}
The multiplication in \eqref{aso} of the form \eqref{smult} is
associative if and only if $E:\F\arrow\F$ is an invertible
endomorphism, i.e.
\begin{equation*}
    E(u v) = E(u) E(v).
\end{equation*}
\end{proposition}

The proof is straightforward. The Lie bracket in $\alg$ is given by
the commutator $\brac{A,B} = AB-BA$, where $A,B\in\alg$.

Let $\tr:\alg\arrow\Km$ be a trace form, being a linear map, such
that
\begin{equation}\label{str}
  \tr(A) :=  \dual{\free (A)},
\end{equation}
where $\free (A):= a_0$ for $A=\sum_i a_i\e^i$ and
$\dual{\cdot}$ denotes a functional whose
form depends on the realization of the algebra $\F$ and the
endomorphism $E$. We assume that $\dual{\cdot}$ is such that
\begin{equation}\label{rest}
    \dual{Ef} = \dual{f}\qquad f\in\F
\end{equation}
holds. Notice that from this assumption it follows that $\dual{uEv}
= \dual{vE^{-1}u}$, thus the adjoint of operator $E$ is $E^\dag =
E^{-1}$.

\begin{proposition}\label{sinner}
The bilinear map defined as
\begin{equation}\label{dms}
\bra{A,B}_\alg := \tr\bra{A B}
\end{equation}
is an inner product on $\alg$ which is nondegenerate, symmetric and
$\ad$-invariant.
\end{proposition}
\begin{proof}
The nondegeneracy of \eqref{dms} is obvious. The symmetricity
follows from the definition by using \eqref{rest}. The
$\ad$-invariance is a consequence of the associativity of
multiplication operation in $\alg$, see Lemma \ref{adtr}.
\end{proof}

Let us mention now a few standard realizations of the algebra of
shift operators~\eqref{aso}:
\begin{itemize}
  \item The first one is given by the shift operators on a {\it discrete lattice}. In this case the dynamical
functions are $u_i:\Z\arrow\Km$ and $E^mu(n) := u(n+m)$. The form of the
functional from \eqref{str} is $\dual{f}:= \sum_{n\in\Z}f(n)$.
  \item The second realization is given by the shift operators
on a {\it continuous lattice}. Now $u_i:\Omega\arrow\Km$ are smooth
functions of $x$, where $\Omega = \Si$ or $\Omega = \Rm$, if $u_i$ are
from the Schwartz space. In this case the shift operator is $E^mu(x) :=
u(x+m\hk)$, where $\hk$ is a parameter. In this case the
functional is just integration, i.e.,
$\dual{f}:= \int_\Omega f(x)dx$.
  \item The third realization
is for {\it$q$-discrete} functions $u_i:\Km_q\arrow\Km$, where
$\Km_q := q^{\Z}\cup \{0\}$ for $q\neq 0$. Here
$E^mu(x):=u(q^mx)$ and $\dual{f}:= \sum_{n\in\Z}f(q^n)$.
\end{itemize}

All the above realizations lead to the construction of lattice
soliton systems: discrete, continuous and $q$-discrete,
respectively. More realizations can be made by means of the discrete
one-parameter groups of diffeomorphisms \cite{Blaszak9} or by means
of jump operators on time scales \cite{G-G-S,bss}.
Notice that in the continuous limit the algebra of shift operators
on lattices gives the Poisson algebra for \eqref{pb} with $r=1$.
Thus, the quasi-classical limit of lattice soliton systems are respective
dispersionless systems. The situation in the $q$-discrete case is similar \cite{Blaszak9}.

The subspaces $\alg_{\me k-1}$ and $\alg_{< k-1}$ of \eqref{aso} are
Lie subalgebras only for $k=1$ and $k=2$ and the classical
$R$-matrices following from the decomposition of $\alg$ are
\eqref{rp} $R = P_{\me k-1} - \frac{1}{2}$. Their adjoints with respect to the above inner product are given by $R^*= P_{<2-k} - \frac{1}{2}$, respectively.

As a result, we have two Lax hierarchies:
\begin{equation}\label{laxh}
L_{t_n} = \brac{\bra{L^{\frac{n}{N}}}_{\me k-1},L}= \pi_0 dH_n = \pi_1
dH_{n-1}\qquad k=1,2,
\end{equation}
of infinitely many mutually commuting systems. Let \eqref{laxh} be
generated by powers of appropriate Lax operators $L\in \alg$ of the
form
\begin{align}\label{l}
L = u_{N}\e^N + u_{N-1}\e^{N-1} + u_{N-2}\e^{N-2} + u_{N-3}\e^{N-3}
+ ... ,
\end{align}
where $u_{N}=1$ for $k=1$.

The bi-Hamiltonian structure of the Lax hierarchies \eqref{laxh} is
defined by the compatible (for fixed $k$) Poisson tensors given by
the formulas \eqref{pot1},
\begin{align*}
\pi_0 dH = \brac{L,(dH)_{<k-1}}+\bra{\brac{dH,L}}_{< 2-k}
\end{align*}
and
\begin{align}
\nonumber \pi_1 dH &= \frac{1}{2}\bra{\brac{L,\bra{LdH + dHL}_{<k-1}}
+ L\bra{\brac{dH,L}}_{<2-k}+\bra{\brac{dH,L}}_{<2-k}L}\\
\label{sqt} &\qquad\quad + (2-k)\brac{(E+1)(E-1)^{-1}\ \free
\bra{\brac{dH,L}},L},
\end{align}
where the operation $(E-1)^{-1}$ is the formal inverse of $(E-1)$.
The second Poisson tensor is a Dirac reduction of \eqref{pot2b} as in this case
$\widetilde{R} = R + (k-\frac{3}{2}P_0)$ satisfies the related
condition.

The differentials $dH(L)$ of functionals $H(L)\in \smf(\alg)$ for
\eqref{l} have the form $dH = \sum_{i=2-k}^\infty\e^{i-N}\var{H}{u_{N-i}}$ and
the respective Hamiltonians \eqref{casq} are
\begin{equation*}
  H_n(L) = \frac{N}{N+n}\tr\bra{L^{\frac{n}{N}+1}}\qquad N\neq -n.
\end{equation*}

The theory of lattice soliton systems of Toda type with Lax
representations given by means of the shift operators was introduced
for the first time by B. Kupershmidt in \cite{kup}. This class of
systems, follows from the classical $R$-matrix formalism applied to
the algebra of shift operators, was investigated from this point
of view in \cite{Blaszak1} and \cite{Oevel4}.

\begin{example}
Consider the case of $k=1$ with \eqref{l} normalized as $L = \e +
\sum_{i=0}^\infty u_{i}\e^{-i}$. The first chain from the Lax
hierarchy \eqref{laxh} has the form
\begin{align}
  L_{t_1} = \brac{\bra{L}_{\me 0},L}=\pi_0dH_1=\pi_1dH_0\iff
(u_i)_{t_1} = (E-1)u_{i+1}+u_i(1-E^{-i})u_0,
\end{align}
and its explicit bi-Hamiltonian structure is given by
\begin{align*}
\pi_0^{ij} &= E^ju_{i+j}-u_{i+j}E^{-i}\\
\pi_1^{ij} &=
\sum_{k=0}^{i}\brac{u_kE^{j-k}u_{i+j-k}-u_{i+j-k}E^{k-i}u_k +
u_i\bra{E^{j-k}-E^{-k}}u_j}\\
&\qquad + u_i\bra{1-E^{j-i}}u_j +
E^{j+1}u_{i+j+1}-u_{i+j+1}E^{-i-1}
\end{align*}
together with the Hamiltonians $H_0 = \dual{u_0}$ and $H_1 = \dual{u_1+\frac{1}{2}u_0^2}$.
\end{example}

Assume now that the dynamical fields depend on an additional variable $y\in\Si$. Then after the central
extension procedure with the Maurer-Cartan two-cocycle \eqref{mc}
the (2+1)-dimensional Lax hierarchy takes  the form
\eqref{ceh}, \cite{Blaszak3},
\begin{equation*}
    L_{t_n} = \brac{\bra{dC_n}_{\me k-1},L-\alpha\,\pr_y} = \pi_0 dC_n\qquad k=1,2,
\end{equation*}
for the Lax operator \eqref{l} with $N>1$,
where $dC_n =\sum_{i=0}^\infty a_{n-i}\e^{n-i}$
are solutions of \eqref{novlax}, $\brac{dC_n,L-\alpha\,\pr_y} = 0$.
The Poisson tensor \eqref{ept2} is given by
\begin{align*}
\pi_0 dH =
\brac{L-\alpha\,\pr_y,(dH)_{<k-1}}+\bra{\brac{dH,L-\alpha\,\pr_y}}_{<
2-k},
\end{align*}
and the Dirac reduction is not required.

\begin{example}
The case of $k=1$. An example of finite field reduction. The Lax
operator is given by $L = \e + u + v\e^{-1}$.
Then for $(dC_1)_{\me 0} = \e + u$  we have
\begin{align*}
L_{t_1} = \brac{(dC_1)_{\me 0},L-\alpha\pr_y} = \pi_0 dH_1 = \pi_1
dH_0\iff \pmatrx{cc}{u\\ v}_{t_1} = \pmatrx{cc}{\bra{E-1}v + \alpha u_y\\
v\bra{1-E^{-1}}u}.
\end{align*}
The respective Poisson tensors are
\begin{equation*}
 \pi_0 = \pmatrx{cc}{\alpha\pr_y & \bra{E-1}v\\
v\bra{1-E^{-1}} & 0}\ {\rm and, if}\,\alpha=0,\ \pi_1 = \pmatrx{cc}{ Ev -vE^{-1} & u\bra{E-1}v\\
v\bra{1-E^{-1}}u & v\bra{E-E^{-1}}v}.
\end{equation*}
The Hamiltonians are $H_1 = \dual{v+\frac{1}{2}u^2}$ and $H_0 = \dual{u}$.
\end{example}

\subsection{Algebra of $\delta$-pseudo-differential operators}

Consider a generalized derivative in the algebra $\F$ of dynamical
fields with values in $\Km$ given by a linear map $\Delta:\F\arrow\F$
that satisfies the generalized Leibniz rule
\begin{equation}\label{gder}
    \Delta(uv) = \Delta(u)v + E(u)\Delta
\end{equation}
for an algebra endomorphism $E:\F\arrow\F$. If $E=1$,
then $\Delta$ is an ordinary derivative. According to \eqref{gder}
we define a generalized differential operator
\begin{equation}\label{delta}
    \delta u = \Delta(u) + E(u)\delta.
\end{equation}
In \eqref{delta} and in all subsequent expressions $\Delta$ and $E$
act only on the nearest $\Delta$-smooth function on their right-hand side.
As above, the expression \eqref{delta} is a counterpart of an
ordinary differential operator $\pr$ such that $\pr u = u_x +u\pr$.
Using \eqref{delta} we have
\begin{equation*}
  \delta^{-1}u = E^{-1}u\delta^{-1} + \delta^{-1}\Delta^\dag u\delta^{-1}
= E^{-1}u\delta^{-1} + E^{-1}\Delta^\dag u\delta^{-2}+ E^{-1}{\Delta^\dag}^2 u\delta^{-3} + \ldots ,
\end{equation*}
where $\Delta^\dag := -\Delta E^{-1}$.

Hence, we can further define an algebra of
$\delta$-pseudo-differential operators
\begin{equation}\label{algebra}
 \alg= \alg_{\me k}\oplus \alg_{< k} = \pobr{\sum_{i\me k}u_{i}\delta^{i}}\oplus
 \pobr{\sum_{i<k}u_{i}\delta^{i}},
 \end{equation}
where $u_i\in\F$. The above algebra is noncommutative and
associative. In fact we have the following result.

\begin{proposition}
The algebra \eqref{algebra} generated by the rule of the form
\eqref{delta} is associative if and only if $\Delta:\F\arrow\F$
satisfies \eqref{gder} and $E:\F\arrow\F$ is an algebra
endomorphism.
\end{proposition}

We omit the proof, which can be done by induction, and it suffices
to consider the associativity condition, $(AB)C = A(BC)$, only for
$A=\delta^i$, $B=b\delta^j$ and $C=c$, where $a,b,c\in\F$.

In fact we will consider only two most important special cases of
generalized derivatives and $\delta$-pseudo-differential operators:
\begin{itemize}
  \item The first case is the ordinary derivative, i.e., \eqref{gder} with $E=1$. Thus let $\F$ consist
  of smooth functions $u_i:\Omega\arrow\Km$ of $x$, where $\Omega = \Si$ or
  $\Omega = \Rm$ (if $u_i$ belong to the Schwartz space). Thus $\Delta = \pr_x$, and we let $\delta := \pr$. In this case $\alg$ is the algebra of standard pseudo-differential
operators with the trace form given by
\begin{equation}\label{str1}
 \tr A = \dual{\res A} = \int_\Omega \res A\,dx\qquad A=\sum_ia_i\pr^{i}\in\alg, \quad \res A:= a_{-1}.
\end{equation}
  \item The second case is the generalized derivative given by a difference operator.
Thus, let
\begin{equation*}
    \Delta = \frac{1}{\mu}\bra{E-1}\quad\Longrightarrow\quad \delta = \frac{1}{\mu}\bra{\e-1},
\end{equation*}
where $\mu$ is constant, $E$ is an algebra endomorphism and $\e$ is a
shift operator such that \eqref{smult} holds, see the realizations from the previous
section. In particular, in the continuous lattice case:
\begin{equation*}
   \Delta f(x)= \frac{f(x+\hk)-f(x)}{\hk},\quad Ef(x)=f(x+\hk),\quad \mu=\hk;
\end{equation*}
and in the $q$-deformed case:
\begin{equation*}
     \Delta f(x)= \frac{f(qx)-f(x)}{q-1},\quad Ef(x)=f(qx),\quad \mu=q-1.
\end{equation*}
Here the trace form is exactly the same as in the case of shift operators algebra
from previous section \eqref{str}. Thus, we have to consider the
restriction \eqref{rest} again. Notice that
$\delta$-pseudo-differential operators can be represented uniquely by shift
operators, with the convention that the $\delta$-operators of negative orders are
expanded into shift operators of negative orders as well. Note that in this
case $\Delta^\dag = -\Delta E^{-1}$.
\end{itemize}

In the quasi-classical limit the algebra of standard pseudo-differential operators
gives the Poisson algebra with canonical Poisson bracket \eqref{pb} for $r=0$.
Thus, in dispersionless limit continuous soliton systems yields respective
dispersionless systems. In the case when difference operator $\Delta$ is given on lattices, the continuous limit of the algebra of pseudo-$\delta$-differential operators is actually given by the algebra of standard pseudo-differential operators. Thus, in contrary to the previous case, the quasi-classical limit
of discrete soliton systems yields continuous soliton systems. This situation is similar in the $q$-deformed cases \cite{bss}.

The reason why we consider both cases, i.e., continuous and discrete,
is the fact that they can be unified into a single consistent scheme
using the so-called time scales \cite{G-G-S,bss,sbs}.
This scheme also includes soliton systems with spatial variable belonging
to the space being a composition of continuous and discrete intervals.
In \cite{sbs} the trace formula on the algebra of pseudo-$\delta$-differential operators
is unified to the form that covers all the above realizations as well the cases with nonconstant
$\mu$.

\begin{proposition}
The inner product \eqref{sc} on $\alg$ defined by means of traces
\eqref{str1} and \eqref{str} in both cases is nondegenerate,
symmetric and $\ad$-invariant.
\end{proposition}
\begin{proof}
In the first case of ordinary derivative nondegeneracy is obvious,
the symmetricity follows from integration parts and the assumption
that integrals of total derivatives vanish. The $\ad$-invariance is a
consequence of Lemma \ref{adtr}. In the second case of difference
operators the proof follows from Proposition \ref{sinner} and the
fact that $\delta$-pseudo-differential operators can be expanded by
means of shift operators.
\end{proof}

The Lie structure on \eqref{algebra} is given by the commutator
$[A,B] = AB-BA$, where $A,B\in\alg$. In general the subspaces
$\alg_{\me k}$ and $\alg_{< k}$ are Lie subalgebras of $\alg$ only
for $k=0$ and $k=1$. However, if $E=1$ they are Lie subalgebras also
for $k=2$. The classical $R$-matrices following from the
decomposition of $\alg$ are $R = P_{\me k} - \frac{1}{2}$. Hence, we
have the following Lax hierarchies of commuting evolution equations
\begin{equation}\label{glh}
       L_{t_n} = \brac{\bra{L^{\frac{n}{N}}}_{\me k},L} = \pi_0 dH_n = \pi_1 dH_{n-1}\qquad k=0,1
\end{equation}
and for $k=2$ if $E=1$. In \eqref{glh} $L\in\alg$. The
infinite-field Lax operators generating \eqref{glh} are given in the
form
\begin{align}\label{lo}
 L=
u_N\delta^N+u_{N-1}\delta^{N-1}+u_{N-2}\delta^{N-2}+u_{N-3}\delta^{N-3}+\ldots\quad
N\me1,
\end{align}
where in general for $k=0$ the field $u_N$ is time-independent; if
$E=1$ then for $k=0$ the field $u_{N-1}$ is also time-independent
and for $k=1$ the same for the field $u_N$.

The explicit form of the differentials $dH = \sum_i\delta^{-i-1}\gamma_i$ with respect to general Lax operator \eqref{lo} has to be such that \eqref{eu} is valid. See the Example \ref{ex}. In the case of $E=1$ we have $\gamma_i = \var{H}{u_i}$. The Hamiltonians \eqref{casq} are given by
\begin{equation*}
  H_n(L) = \frac{N}{N+n}\tr\bra{L^{\frac{n}{N}+1}}\qquad N\neq -n.
\end{equation*}

We will consider the Hamiltonian structures of \eqref{glh} only for
$k=0$. Thus, for $k=0$ we have $R^\star = -R$. Hence the linear Poisson
tensor is given by \eqref{pot1},
\begin{align}\label{lin1}
\pi_0 dH = \brac{(dH)_{\me 0},L} - \bra{\brac{dH,L}}_{\me 0}.
\end{align}
The case of quadratic Poisson tensor is more complex. For the pseudo differential
operators, when $E=1$, the quadratic Poisson tensor is given by
\begin{equation}\label{quad2}
\pi_1^{red}dH = \bra{L dH}_{\me 0} L - L \bra{dH L}_{\me 0} + \frac{1}{N}\brac{\pr_x^{-1}\bra{\res \brac{dH,L}},L},
\end{equation}
after Dirac reduction applied to \eqref{pot2b}.
In the case of difference operators ($E\neq 1$), as the
decomposition for $k=0$ of the algebra of
$\delta$-pseudo-differential operators \eqref{algebra} coincides with
the decomposition  of the algebra of purely shift operators
\eqref{aso} for $k=1$, the quadratic Poisson tensor is given by
\eqref{sqt}, see \cite{Oevel4,sbs} for further details.

The pseudo-differential case ($E=1$) of \eqref{glh} with  $k=0$ and
finite-field reductions, given by  constraint $u_i=0$ for $ i<0$, is
the well-known Gelfand-Dickey hierarchy \cite{Gelfand}. More details
of the $R$-matrix formalism applied to the algebra of
pseudo-differential operators for the remaining values of
$k$ can be found in \cite{Konopelchenko} or \cite{Blaszak2}.

\begin{example}
Consider the infinite-field case $k=0$ of \eqref{lo} in the
ordinary derivative case $E=1$, i.e., $L = \pr + \sum_{i=0}^\infty
u_i\pr^{-i-1}$. This is the case of the well-known KP hierarchy~\cite{Dickey}.
Then we find the first nontrivial dispersive chain from \eqref{dlaxh},
\begin{align*}
 L_{t_2} &= \brac{\bra{L^2}_{\me 0},L} = \pi_0 dH_2 = \pi_1 dH_1\iff\\
  &\qquad\quad(u_i)_{t_2} = (u_i)_{2x} + 2(u_{i+1})_x - 2\sum_{k=1}^i(-1)^k\binom{i}{k}u_{i-k}(u_0)_{kx},
\end{align*}
where $(L^2)_{\me 0} = \pr^2 + 2u_0$. The bi-Hamiltonian structure is given by linear
\begin{align*}
\pi_0^{ij} = \sum_{k=1}^{j}\binom{j}{k}\pr_x^ku_{i+j-k}-\sum_{k=1}^{i}(-1)^{k}\binom{i}{k}u_{i+j-k}\pr_x^k
\end{align*}
and quadratic Poisson tensor
\begin{align*}
  \pi_1^{ij} &= \sum_{k=1}^{j+1}\binom{j+1}{k}\pr_x^ku_{i+j-k+1}-\sum_{k=1}^{i+1}(-1)^k\binom{i+1}{k}u_{i+j-k+1}\pr_x^k\\
  &\qquad +\sum_{l=0}^{j-1}\brac{\sum_{k=1}^{j-l-1}\binom{j-l-1}{k}u_l\pr_x^ku_{i+j-k-l-1}-\sum_{k=1}^{i}(-1)^{k}\binom{i}{k}u_{i-k+l}\pr_x^ku_{j-l-1}}\\
  &\qquad -\sum_{l=0}^{j-1}\sum_{k=0}^{i}\sum_{s=1}^{j-l-1}(-1)^k\binom{i}{k}\binom{j-l-1}{s}u_{i-k+l}\pr_x^{k+s}u_{j-l-s-1},
\end{align*}
that are obtained from \eqref{lin1} and \eqref{quad2}, respectively. The
respective Hamiltonians are
\begin{align*}
    H_1 = \int_{-\infty}^\infty u_1\,dx\qquad H_2 = \int_{-\infty}^\infty \bra{u_0^2+u_2} dx.
\end{align*}
\end{example}

\begin{example}\label{ex}
The case of $k=0$. In this example we present the associated integrable
systems and their Hamiltonian structures in the form that is valid
in both continuous and discrete cases. In the first case, $E=1$,
$\mu=0$ and $\Delta = -\Delta^\dag = \pr_x$.

Let the Lax operator be given in the form
\begin{equation}\label{a0}
    L = \delta + \mu \psi\varphi + \psi\delta^{-1}\varphi.
\end{equation}
Then the first and the second flows from the Lax hierarchy
\eqref{laxh} are
\begin{equation}\label{a1}
\begin{split}
    \psi_{t_1} &= \mu\psi^2\varphi+\Delta\psi,\\
    \varphi_{t_1} &= -\mu\varphi^2\psi - \Delta^\dag\varphi.
\end{split}
\end{equation}
and
\begin{equation}\label{a2}
 \begin{split}
 \psi_{t_2} &= \mu^2\psi^3\varphi^2 + 2\psi^2\varphi + \Delta^2\psi + \Delta\bra{\mu\psi^2\varphi}
    + 2\mu\psi\varphi\Delta\psi + \mu\psi^2\Delta^\dag\varphi\\
    \varphi_{t_2} &= -\mu^2\psi^2\varphi^3 - 2\psi\varphi^2 - {\Delta^\dag}^2\varphi
    - \Delta^\dag\bra{\mu\psi\varphi^2} - \mu\varphi^2\Delta\psi - 2\mu\psi\varphi\Delta^\dag\varphi .
 \end{split}
\end{equation}

For the Lax operator \eqref{a0} the differential of a functional $H$
such that \eqref{eu} is valid is given by
\begin{equation*}
    dH = \frac{1}{\varphi}\var{H}{\psi} - \frac{1}{\psi}\Delta^\dag\bra{\frac{1}{\varphi}}
    \Delta^{-1}A - \delta\frac{1}{\psi\varphi + \mu\psi\Delta^\dag\varphi}\Delta^{-1}A,
\end{equation*}
where $A = \psi\var{H}{\psi} - \varphi\var{H}{\varphi}$, and $\Delta^{-1}$ is a formal inverse of $\Delta$. The linear and quadratic Poisson tensors take the form \cite{sbs},
\begin{equation*}
    \pi_0 = \pmatrx{cc}{0 & 1\\ -1 & 0},\quad
    \pi_1 =
    \pmatrx{cc}{-\mu\psi^2 - 2\psi\Delta^{-1}\psi & \Delta + 2\mu\psi\varphi + 2\psi\Delta^{-1}\varphi\\
     -\Delta^\dag + 2\varphi\Delta^{-1}\psi &  -\mu\varphi^2 - 2\varphi\Delta^{-1}\varphi },
\end{equation*}
while the Hamiltonians are
\begin{align*}
    H_0 &= \dual{\psi\varphi},\quad H_1 = \dual{\frac{1}{2}\mu\psi^2\varphi^2 + \varphi\Delta\psi},\\
    H_2 &= \dual{\frac{1}{3}\mu^2\psi^3\varphi^3 + \psi^2\varphi^2 + \varphi\Delta^2\psi
       + \mu\psi\varphi^2\Delta\psi + \mu\psi^2\varphi\Delta^\dag\varphi}.
\end{align*}

In particular, when $E=1$, $\mu=0$ and $\Delta = \pr_x$ the above
bi-Hamiltonian hierarchy is precisely the bi-Hamiltonian field soliton
AKNS hierarchy \cite{Oevel2}. In this case the first nontrivial flow
is the second one \eqref{a2}, i.e., the AKNS system. When $\Delta$ is the difference
operator, we obtain in particular the lattice \cite{Oevel4} and the $q$-discrete counterparts of
the AKNS hierarchy, where the first nontrivial flow is \eqref{a1}.

In fact this example is more general and also includes more complex situations
when $\mu$ is non-constant (time-independent) function on $\Rm$, for the details we send
the reader to~\cite{sbs}.
\end{example}

Let $\alg$ be the algebra of pseudo-differential operators with
the coefficients depending on independent spatial variables $x$ and $y$.
The central extension procedure with \eqref{mc}
yields the following Lax hierarchy \eqref{ceh}\cite{Blaszak3}
\begin{equation*}
    L_{t_n} = \brac{\bra{dC_n}_{\me k},L-\alpha\pr_y} = \pi_0 dC_n,
\end{equation*}
where $k=1,2$ or $k=3$ and $L = u_N\pr^N + u_{N-1}\pr^{N-1} + u_{N-2}\pr^{N-2} + \ldots$,
with $u_N=1,u_{N-1}=0$ for $k=0$ and only $u_N=1$ if $k=1$. The Lax
hierarchies are generated by $dC_n = \sum_{i=0}^\infty
a_{n-i}\pr^{n-i}$ solving \eqref{novlax}, $\brac{dC_n,L-\alpha\pr_y} = 0$.
The associated Poisson tensor \eqref{ept2} is given by
\begin{align*}
\pi_0 dH = \brac{(dH)_{\me 0},L-\alpha\pr_y} -
\bra{\brac{dH,L-\alpha\pr_y}}_{\me 0},
\end{align*}
and requires no Dirac reduction.

\begin{example}
The case of $k=0$ with the Lax operator of the form $L = \pr^2 + u$.
Then, for $(dC_3)_{\me 0} = \pr^3 + \frac{3}{2}u\pr + \frac{3}{4}\bra{u_x + \alpha \pr_x^{-1}u_y}$
we obtain the (2+1)-dimensional KP equation
\begin{align*}
  L_{t_3} = \brac{(dC_3)_{\me 0},L-\alpha \pr_y}=\pi_0 dH_3 = \pi_1
dH_1\iff u_{t_3} = \frac{1}{4}u_{xxx} + \frac{3}{2}uu_x +
\frac{3}{4}\alpha^2\pr_x^{-1}u_y,
\end{align*}
where the Poisson tensors are $\pi_0 = 2\pr_x$ and, if $\alpha=0$,
$\pi_1 = \frac{1}{2}\pr_x^3+\pr_xu+u\pr_x$. The Hamiltonians are
\begin{equation*}
H_1 = \iint_{\Omega \times \Si}\frac{1}{4}u^2\,dxdy\quad H_3 =
\iint_{\Omega\times\Si}\frac{1}{16}\bra{2u^3+uu_{xx}+3\alpha^2u\pr_x^{-2}u_{yy}}dxdy.
\end{equation*}
\end{example}

More examples of $(2+1)$-dimensional field systems can be found in \cite{Blaszak3,Szablikowski2}, where one can find also systems that are purely $(2+1)$-dimensional phenomena.

\subsection{$sl(2,\Cm)$ loop algebra}\label{sl2}

We will follow the scheme from Section \ref{la}. Consider a
loop algebra over the classical Lie algebra $\alg=sl(2,\Cm)$ of
traceless $2\times2$ nonsingular matrices, i.e., $\alg^\la =
sl(2,\Cm)\bc{\la,\la^{-1}}$, with coefficients being smooth
dynamical functions of variable $x\in\Omega$. The space
$u_i:\Omega\arrow\Km$ are smooth functions of $x$, where $\Omega =
\Si$ if we assume these functions to be periodic in $x$ or $\Omega =
\Rm$ if these functions belong to the Schwartz space. The commutator
defines the Lie bracket in $sl(2,\Cm)$ and readily extends to
$\alg^\la$ since \eqref{liel}.

We already know that there are two natural decompositions of
$\alg^\la$ into Lie subalgebras \eqref{ls}, and we consider only the
one for $k=0$ yielding the classical $R$-matrix $R= P_+ - \frac{1}{2}$,
with $P_+$ being the projection onto the nonnegative powers of $\la$.

The trace form on $\alg^\la$ is  \eqref{ltr}, $\Tr (a) = \int_\Omega \res\,\tr (a)\, dx$,
where $\res \sum_i a_i\la^i = a_{-1}$ and $\tr$ is the standard trace
of matrices. As the matrices in $sl(2,\Cm)$ are nonsingular, the
trace $\tr$ defines nondegenerate inner product which is also symmetric
and $\ad$-invariant. These properties extend to $\alg^\la$, and
hence $\alg^{*\la}\cong\alg^\la$ and $\ad^*\equiv \ad$.

Applying the central extension procedure with the Maurer-Cartan
two-cocycle \eqref{mc} ($x\equiv y$) yields the Lax hierarchy given by
\eqref{llh} (we take $\alpha=1$), i.e.,
\begin{equation}\label{llh2}
  L_{t_n} = \brac{\bra{dC_n}_+,L-\pr_x} = \ldots = \pi_{l}dC_{n-l}= \ldots
\end{equation}
for the Lax operators
\begin{equation}\label{llo}
   L = \bu_N\la^N + \bu_{N-1}\la^{N-1} + \ldots + \bu_{1-m}\la^{1-m} + \bu_{-m}\la^{-m}\quad N\me -1,
\end{equation}
where $\bu_i\in sl(2,\Cm)$ and $\bu_N$ is a constant matrix. The Lax
hierarchy \eqref{llh2} is generated by $dC_n = \la^n dC_0$ such
that $dC_0 = \sum_{i=0}^\infty a_i\la^i$, where $a_i\in sl(2,\Cm)$,
satisfies $\brac{dC_0,L-\pr_x} = 0$.

For a given functional $H\in\smf{\alg^\la}$ of \eqref{llo} its
differential has the form
\begin{equation*}
    L = \var{H}{\bu_{-m}}\la^{m-1} + \var{H}{\bu_{1-m}}\la^{m-2} + \ldots + \var{H}{\bu_{N-1}}\la^{-N},
\end{equation*}
where
\begin{equation*}
    \var{H}{\bu} = \pmatrx{cc}{\frac{1}{2}\var{H}{u_{11}} & \var{H}{u_{21}}\\[1mm] \var{H}{u_{12}} &
    -\frac{1}{2}\var{H}{u_{11}}}\quad{\rm for}\quad \bu = \pmatrx{cc}{u_{11} & u_{12}\\ u_{21} & -u_{11}}.
\end{equation*}
The multi-Hamiltonian structure for \eqref{llh2} is given by \eqref{pt1}
\begin{equation}\label{pt2}
    \pi_{l} dH = \brac{\bra{\la^l dH}_+,L-\pr_x} - \la^l\brac{dH,L-\pr_x}_+.
\end{equation}
For the general Lax operators \eqref{llo}, if $N\me l\me -m$ for
$N\me 0$ and  $0\me l\me -m$ for $N=-1$ then the Dirac reduction of
\eqref{pt2} is not required.

\begin{example}
Consider the Lax operator of the form
\begin{equation}\label{ak}
L= \pmatrx{cc}{-i & 0 \\ 0 & i}\la + \pmatrx{cc}{0 & q \\ r & 0},\quad
i=\sqrt{-1}.
\end{equation}
Then we find that
\begin{equation*}
     dC_0 = \pmatrx{cc}{-i & 0\\ 0 & i} + \pmatrx{cc}{0 & q\\ r & 0}\la^{-1} +
    \pmatrx{cc}{-\frac{i}{2}rq & \frac{i}{2}q_x\\ -\frac{i}{2}r_x & \frac{i}{2}rq}\la^{-2}
 + \ldots .
\end{equation*}
Hence, from \eqref{llh2} we obtain the Ablowitz-Kaup-Newell-Segur
(AKNS) hierarchy \cite{abl}
\begin{equation*}
   \pmatrx{cc}{q\\r}_{t_1} = \pmatrx{cc}{q_x\\ r_x},\quad \pmatrx{cc}{q \\r}_{t_2} = \pmatrx{cc}{\frac{i}{2}q_{xx}-irq^2\\ -\frac{i}{2}r_{xx}+ir^2q},\quad
   \pmatrx{cc}{q \\r}_{t_3} = \pmatrx{cc}{-\frac{1}{4}q_{xxx}+\frac{3}{2}rqq_x\\ -\frac{1}{4}r_{xxx}+\frac{3}{2}rqr_x},\ \ldots\ .
\end{equation*}

The bi-Hamiltonian structure is given by the Poisson tensors \eqref{pt2}
\begin{align*}
   \pi_0 = \pmatrx{cc}{0 & -2i\\ 2i & 0}\qquad \pi^{red}_1 = \pmatrx{cc}{2q\pr_x^{-1}q & \pr_x-2q\pr_x^{-1}r\\ \pr_x-2r\pr_x^{-1}q & 2r\pr_x^{-1}r}
\end{align*}
with the hierarchy of Hamiltonians
\begin{equation*}
   H_1 = \int_\Omega \frac{i}{2}q_xr\,dx,\ H_2 = \int_\Omega \frac{1}{4}\bra{q^2r^2-q_{xx}r}dx,\
   H_3 = \int_\Omega \frac{i}{8}\bra{3qq_xr^2-q_{xxx}r}dx,\ \ldots\ .
\end{equation*}
Notice that the need for the Dirac reduction for the Poisson
tensor $\pi^{red}_1$ follows from the fact that \eqref{ak} is not
the most general Lax operator \eqref{llo} from $\alg^\la$ with $N=1$
and $m=0$.

The reduction  $q=\psi,r=\psi^\star$ of the above AKNS hierarchy
yields the non-linear Schr\"odinger hierarchy, while the reductions
$q=iu,r=i$ and $q=r=iv$ give rise to the KdV and the mKdV hierarchies respectively.
\end{example}

\section*{Acknowledgement}

The work was partially supported by Polish MNiSW research grant
no. N N202 404933. B.Sz. was supported (starting from 1 October 2008) by the European Community
under a Marie Curie Intra-European Fellowship, contract no. PIEF-GA-2008-221624.

\end{document}